\documentclass[runningheads,a4paper]{llncs}    

\usepackage{booktabs} 
\usepackage{color}
\usepackage{amsmath} 
\usepackage{amssymb}
\usepackage{graphicx}
\usepackage{dsfont}
\usepackage{amsfonts}
\usepackage{stmaryrd}
\usepackage[all]{xy}
\usepackage{multirow}
\usepackage{qcircuit}
\usepackage{tikz}
\usetikzlibrary{positioning}

\def\>{\ensuremath{\rangle}}
\def\<{\ensuremath{\langle}}

\newcommand {\supp } {{\rm supp}}

\newcommand {\E } {{\mathcal{E}}}

\newcommand{\hs}{\mathcal{H}}

\newcommand {\tr} {{\mathit{tr}}}

\newcommand {\sem}[1]{\ensuremath{\llbracket {#1}\rrbracket_\mathbb{I}}}
\newcommand {\semp}[1]{\ensuremath{\llbracket {#1}\rrbracket_\mathbb{I}^\ast}}

\newtheorem{thm}{Theorem}[section]

\newtheorem{lem}{Lemma}[section]
\newtheorem{defn}{Definition}[section]

\newtheorem{exam}{Example}[section]

\newtheorem{rem}{Remark}[section]

\begin{document}

\mainmatter

\title{\textit{Birkhoff-von Neumann Quantum Logic} as an Assertion Language for Quantum Programs}
\author{Mingsheng Ying}
\institute{State Key Laboratory of Computer Science, Institute of Software, Chinese Academy of Sciences, China\\     
    Department of Computer Science and Technology, Tsinghua University, China\\
\email{yingms@ios.ac.cn; yingmsh@tsinghua.edu.cn}}
\titlerunning{Assertion Language for Quantum Programs}
\authorrunning{Ying}

\maketitle

\begin{abstract} A first-order logic with quantum variables is needed as an assertion language for specifying and reasoning about various properties (e.g. correctness) of quantum programs. 
Surprisingly, such a logic is missing in the literature, and the existing first-order Birkhoff-von Neumann quantum logic deals with only classical variables and quantifications over them.  
In this paper, we fill in this gap by introducing a first-order extension of Birkhoff-von Neumann quantum logic with universal and existential quantifiers over quantum variables. Examples are presented to show our logic is particularly suitable for specifying some important properties studied in quantum computation and quantum information. We further incorporate this logic into quantum Hoare logic as an assertion logic so that it can play a role similar to that of first-order logic for classical Hoare logic and BI-logic for separation logic. In particular, we show how it can be used to define and derive quantum generalisations of some adaptation rules that have been applied to significantly simplify verification of classical programs. It is expected that the assertion logic defined in this paper - first-order quantum logic with quantum variables - can be combined with various quantum program logics to serve as a solid logical foundation upon which verification tools can be built using proof assistants such as Coq and Isabelle/HOL. 
\end{abstract}

\section{Introduction}\label{Intro}

\ \ \ \ \ \ \textbf{Program Logics and Assertion Logics}: A major class of verification techniques for classical programs are based on program logics; in particular, Hoare logic and its various extensions. Program logics are designed for specifying dynamic properties of programs. Usually, a program logic is built upon an assertion logic that is employed to describe static properties of program variables \cite{Apt19}. Certainly, the effectiveness of  these verification techniques comes from a combined power of program logics and assertion logics rather than the sole role of the former. However, this point has often not been seriously noticed. The reason is possibly that first-order logic is commonly adopted as an assertion logic, it is ubiquitous in mathematics, computer science and many other fields, and thus its role is considered for granted and frequently overlooked. The important role of assertion logics in program verification became particularly clear through the great success of separation logic \cite{SL02,SL19}, which enables local reasoning by expanding assertion logic with new connectives (namely separation conjunction and the associated implication) that are not definable in first-order logic \cite{BI01}, especially by adopting the logic BI of bunched implications \cite{BI99} as its assertion language. 

{\vskip 3pt}

\ \ \ \ \textbf{Quantum Hoare Logic}: The rapid progress of quantum computing hardware in the last decade has stimulated recent intensive research on quantum programming methodology. In particular, several program logics have been defined \cite{Jor04,Cha06,BS06,Feng07,Kaku09,Unruh19a,Unruh19b,Barthe20,Kart20,Zhou21,Le22,Li21} and various verification and analysis techniques have been developed \cite{Mary17,Mary19,Vali21,Mary21a,Mary21b,Wu19,Gu21,Gu19,Carbin22,ETH21,Yu21,YDFJ09} for quantum programs (see also surveys \cite{Ying19,Survey_UK,Survey_FR}). Among them, D'Hondt and Panangaden \cite{DP06} introduced the notion of quantum weakest precondition, where a quantum predicate is considered as a physical observable with eigenvalues in the unit interval, which can be mathematically modelled as a Hermitian operator between the zero and identity operators, and is often called an effect in the quantum foundations literature. A (relatively) complete quantum Hoare logic (QHL for short) with such quantum predicates was established in \cite{Ying19}. Moreover, a QHL theorem prover was implemented based on Isabelle/HOL for verification of quantum programs \cite{Liu19}.        

A major hurdle for the applicability of the current version of QHL to verification of large quantum algorithms comes from the poor expressivity of its assertion language. 
To see this more clearly, let us compare the assertions used in classical Hoare logic and QHL. A classical assertion for a program is a predicate, i.e. a Boolean-valued function, over the state space of the program. First-order logic used as an assertion language enables that every assertion is represented by a logical formula constructed from atomic formulas using propositional connectives and universal and existential quantifications. It enhances the applicability of Hoare logic in at least two ways: (i) a logical representation of an assertion is often much more economic than as a Boolean-valued function over the entire state space; (ii) first-order logic can be used to infer entailment between assertions, which can help us to apply the rules of Hoare logic more efficiently. In contrast, following \cite{DP06}, a quantum predicate in QHL is currently described as a Hermitian operator on the Hilbert space of quantum variables, which is, for example, a $2^n\times 2^n$ matrix for the case of $n$ qubits. In a sense, this can be seen as a quantum counterpart of Boolean-valued function representation of a classical predicate. Then a verification condition for a quantum program with $n$ qubits is derived in a QHL prover \cite{Liu19} as an inequality between two $2^n\times 2^n$ matrices or equivalently the semi-definite positivity of their difference, which is hard to check when dealing with large quantum algorithms because the size of the involved matrices grows up exponentially as the number of program variables. 
This is very different from a verification condition for a classical program, which is written as a first-order logical formula and can be inferred from validity of its sub-formulas structurally using logical rules for connectives and quantifiers.  
\textit{So, scalable applications of QHL requires an assertion logic that can play a role similar to that of first-order logic for classical Hoare logic?} 

The proof assistant Coq has been very successfully used in building verification tools for quantum compilers as well as quantum algorithms \cite{Mary17,Mary19,Mary21a,Mary21b,Peng22}. Although these tools are not based on quantum Hoare logic, we expect that they can also be empowered by an assertion logic with the same benefits as discussed above. 

{\vskip 3pt}

\ \ \ \ \textbf{Runtime Assertion Checking in Quantum Computing}:  The need of a logical language for quantum assertions also arises in another line of research. As is well-known, runtime assertion checking is one of the most useful automated techniques in classical software testing and analysis for detecting faults and providing information about their locations \cite{Clarke06}. This technique has recently been extended to quantum computing. The first assertion scheme was proposed in \cite{Huang19} to check whether a program variable is in a given quantum state. This scheme is essentially statistical and cannot be implemented dynamically at runtime because the destructive measurements used there may cause collapse of quantum states. A runtime assertion scheme was then introduced in \cite{Liu20} by employing non-destructive measurement and SWAP test. The assertions in \cite{Huang19,Liu20} specify that a program variable is in a single state (or multiple variables are in an entangled state, which is still a single state). The scheme in \cite{Liu20} was further generalised in \cite{Liu21} to assert that program variables are in one of several quantum states, say $|\psi_1\rangle,...,|\psi_n\rangle$. Unfortunately, lacking a precisely defined assertion language leads to some inaccuracy and even incorrectness there, for example the confusion between (i) a program variable is in one of quantum state $|\psi_1\rangle,...|\psi_n\rangle$, or in set $S=\{|\psi_1\rangle,...|\psi_n\rangle\}$, which is a proposition in classical logic;
and (ii) a program variable is in the subspace $X$ spanned by $S$, which is a proposition in Birkhoff-von Neumann quantum logic \cite{BvN36}. This confusion can also be seen as a misunderstanding of the different interpretations of connective \textquotedblleft $\vee$ (or)\textquotedblright\ in classical logic and quantum logic: if we write $\alpha_i$ for the proposition of being in state $|\psi_i\rangle$, then set $S$ is the semantics of $\bigvee_{i=1}^n\alpha_i$ in classical logic, and subspace $X$ is the semantics of $\bigvee_{i=1}^n\alpha_i$ in quantum logic (see Example \ref{q-connectives} for more detailed discussion). \textit{Therefore, one can expect that a logical language for specifying quantum assertions can help to prevent these slips.}     

{\vskip 3pt}

\textbf{Quantum Logic with Classical Variables}: 
Quantum logic (QL for short) has been developed for about 80 years since Birkhoff and von Neumann's seminal paper \cite{BvN36} to provide an appropriate logic for reasoning about quantum mechanic systems. 
Then a natural question is: \textit{can Birkhoff-von Neumann quantum logic be directly used as an assertion language for quantum programs?} To answer this question, let us first briefly review the basic ideas of quantum logic. 
It was identified in \cite{BvN36} that a proposition about a quantum system should be mathematically represented by a closed subspace of (or equivalently, a projection operator on) the state space of the system, which is a Hilbert space $\hs$ according to the postulates of quantum mechanics. The early research on quantum logic had been focusing on its algebraic aspect, namely understanding the algebraic structure of the set $S(\hs)$ of closed subspaces of $\hs$, and it was proved that $S(\hs)$ equipped with orthocomplement $^\bot$ and intersection $\cap$ is an orthomodular lattice. The later research naturally turned to a more logic flavour \cite{Dalla04}; that is, quantum logic is defined as a logic with truth values as elements of $S(\hs)$ or even an abstract orthomodular lattice. More precisely, the language of propositional quantum logic is the standard propositional language, but the truth value of each propositional variable is taken from $S(\hs)$, and logical connectives $\neg, \wedge$ are interpreted as $^\bot, \cap$, respectively.  
Furthermore, first-order quantum logic has the standard first-order language with, say, individual variables $x,y,z,...$, function symbols $f,g,...$, predicate symbols $P,Q,...$, connectives $\neg,\wedge$ and quantifier $\forall$. The individual variables $x,y,z,...$ are still classical variables with values taken from a domain $D$. Then function symbols $f,g,...$ are interpreted in the same way as in classical logic, and quantifier $\forall$ binds classical variables. The only difference between the quantum and classical logics is that 
the former is $S(\hs)$-valued; that is, an $n$-ary predicate symbol, say $P$, is interpreted as a mapping $D^n\rightarrow S(\hs)$, and connectives $\neg,\wedge$ are interpreted as $^\bot$ and $\cap$, respectively. 

{\vskip 3pt}

\textbf{Why A Logic with Quantum Variables?} The existing first-order quantum logic (as described above) only has classical variables and seems not a desirable assertion language for quantum programs (and more broadly, logical tool for formal reasoning in quantum computation and quantum information). Instead, a first-order logic with quantum variables whose values are quantum states rather than classical ones will be often more convenient, in particular in analysis and verification of quantum programs and quantum cryptographic and communication protocols. To see this, let us consider a simple example: 


\begin{exam}[Inputs/Outputs of quantum circuit]\label{example-1} The quantum circuit in Figure \ref{combinational-quantum-circuit}
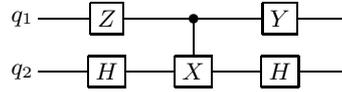
\begin{figure}
\centerline{
\Qcircuit @C=1em @R=0.9em {
q_1\ \ \ \ &\qw  & \gate{Z}  &\qw    &\ctrl{1}   &   \qw  & \gate{Y}  &\qw&\qw\\
q_2\ \ \ \ &\qw  & \gate{H}  &\qw    &\gate{X}   &   \qw  & \gate{H}  &\qw&\qw\\
}
}
    \caption{A quantum circuit.}
    \label{combinational-quantum-circuit}
\end{figure}
can be written as a term with quantum variables:
\begin{equation}\label{term-exam}\tau=Z(q_1)H(q_2)C(q_1,q_2)Y(q_1)H(q_2)\end{equation} where $q_1,q_2$ are qubit variables, operation symbols $Z,Y,H,C$ denote Pauli gates $Z, Y$, Hadamard gate, and controlled-Not, i.e. CNOT, respectively. It is easy to see that whenever in the input to the circuit, $q_1$ is in basis state $|0\rangle$, then the state of $q_2$ in the output is the same as in the input. This fact can be expressed by the first-order formula:
\begin{equation}\label{term-exam1}\beta=(\forall q_1)(\forall q_2)[P_0(q_1)\wedge P(q_1,q_2)\rightarrow P(\tau)]\end{equation} where $P_0$ is a predicate symbol for a single qubit denoting the one-dimensional space spanned by $|0\rangle$, and $P$ is a predicate symbol for two qubits denoting a subspace of the form $\hs_2\otimes X$ with $\hs_2$ being the state space of the first qubit and $X$ an arbitrary subspace of the state space of the second qubit. 
\end{exam}

Quantum variables $q_1,q_2$ and universal/existential quantifications over them appear in a natural way in the above example. But surprisingly, \textit{a first-order quantum logic with quantum variables is still missing in the literature}. The \underline{first contribution} of this paper is to fill in this gap and to define such a logic $\mathcal{QL}$ in which logical formulas can contain quantum variables as their individual variables, and thus 
 (\ref{term-exam1}) is an eligible logical formula. $\mathcal{QL}$ is designed to be a first-order logic with equality $=$ so that we can use it to reason about equality of two quantum states and equivalence of two quantum circuits.

{\vskip 3pt}

\textbf{Quantifications over Quantum Variables}: The introduction of quantum variables into a first-order logic leads us to a fundamentally new issue --- quantification over quantum variables. To see this, let us consider the universal quantification. The existential quantification is similar. Recall that in classical first-order logic, a universally quantified formula $(\forall x)\beta$ is interpreted by  
\begin{equation}\label{classical-qunantifier}(\mathbb{I},v)\models (\forall x)\beta\ {\rm iff}\ (\mathbb{I},v[a/x])\models\beta\ {\rm for\ any\ possible\ value}\ a\ {\rm of}\ x\end{equation}
where $\mathbb{I}$ is an interpretation of the first-order language under consideration, $v$ is a valuation function of individual variables, and $v[a/x]$ is the valuation function that coincides with $v$ for all variables $y\neq x$ but takes the value $a$ for variable $x$. A natural quantum extension of universal quantification is a formula $(\forall \overline{q})\beta$ with a sequence $\overline{q}$ of quantum variables that is interpretd by 
\begin{equation}\label{quantum-quantifier}(\mathbb{I},\rho)\models (\forall\overline{q})\beta\ {\rm iff}\ (\mathbb{I},\E(\rho))\models\beta\ {\rm for\ any}\ \E\in\mathcal{O}_{\overline{q}}\end{equation}
where $\mathcal{O}_{\overline{q}}$ is a set of quantum operations that are allowed to perform on $\overline{q}.$ 

Quantification over quantum variables is much more intrigues than that over classical variables. We note that the quantification in (\ref{quantum-quantifier}) is defined over a sequence of quantum variables rather than over a single variable as usual in classical logic. The reason is that in general, a joint quantum operation on several variables cannot be implemented by a series of local operations on a single variable. Moreover, although we only consider the quantification defined in (\ref{quantum-quantifier}) by allowed operations in logic $\mathcal{QL}$, quantification over quantum variables can be defined in several other different ways that reflect some characteristic features of quantum systems, as will be briefly discussed at the end of this paper. 

{\vskip 3pt}

\textbf{$\mathcal{QL}$ as An Assertion Language for Quantum Programs}: After establishing logic $\mathcal{QL}$, the \underline{second contribution} of this paper is to incorporate $\mathcal{QL}$ into QHL (quantum Hoare logic) so that program logic QHL and assertion logic $\mathcal{QL}$ can properly work together in verification and analysis of quantum programs. In particular, with the help of $\mathcal{QL}$, a series of auxiliary and adaptation rules can be defined and derived for more convenient applications of QHL.   

QHL was first developed in \cite{Ying11} for Hoare triples with pre/postconditions being general quantum predicates represented as Hermitian operators between the zero and identity operators, i.e. effects. A variant of QHL is derived in \cite{Zhou19} by using a special class of quantum predicates, namely projections, as pre/postconditions. It was shown that this variant can be used for simplifying verification of some quantum algorithms, for example HHL (Harrow-Hassidim-Lloyd) algorithm for solving systems of linear equations and qPCA (quantum Principal Component Analysis), as well as runtime assertion checking and testing of quantum programs \cite{Li20}. Correspondingly, there are mainly two variants of QL (quantum logic). The original QL was defined in \cite{BvN36} for logical propositions interpreted as projections (or equivalently, closed subspaces of the Hilbert space of the quantum system under consideration), and now it is often called sharp QL. Unsharp QL was introduced in the framework of effect-based formulation of quantum theory \cite{Kraus}, where logical propositions are interpreted as effects, i.e. quantum predicates as defined in \cite{DP06}. In this paper, we choose to focus on the extension $\mathcal{QL}$ of \textit{sharp QL} as an assertion language for the simplified version of QHL with projections as pre/postconditions. The main ideas and results of this paper will be generalised to the case of \textit{unsharp QL} and QHL with general quantum predicates in a companion paper. 

{\vskip 3pt}

\textbf{Adaptation Rules for Quantum Programs}: As is well-know, there are usually two types of proof rules in a classical program logic (e.g. Hoare logic or separation logic), namely \textit{construct rules} and \textit{adaptation rules} \cite{Apt09,Apt19}. A construct rule is defined for reasoning about correctness of the program construct under consideration. The construct rules make syntax-directed program verification possible. On the other hand, an adaptation rule derives a correctness formula $\{\beta^\prime\} S\{\gamma^\prime\}$ of a program $S$ from an already established correctness formula $\{\beta\}S\{\gamma\}$ of the same program. Such a rule enable us to adapt correctness $\{\beta\}S\{\gamma\}$ to a new context. Adaptation rules can often help us to simplify program verification significantly.   

It is naturally desirable to extend these adaptation rules for verification of quantum programs. Indeed, one of them, namely the consequence rule can be straightforwardly generalised to the quantum case (see e.g. \cite{Ying16}, rule (R.Or) in Figure 4.2). However,  
the pre/postconditions $\beta^\prime$ and $\gamma^\prime$ in the conclusions of other adaptation rules are formed from the pre/postconditions $\beta,\gamma$ in their premises using logical connectives and quantifiers. They have not been generalised to the quantum cases in the previous works due to the lack of proper logical tools. Now $\mathcal{QL}$ provides with us the necessary logical tools, and the \underline{third contribution} of this paper is to derive a series of useful adaptation rules for quantum programs with the help of $\mathcal{QL}$. In particular, quantifiers over quantum variables introduced in $\mathcal{QL}$ are essential in deriving the quantum generalisations of the $\exists$-introduction rule and Hoare adaptation rule which plays a crucial role in reasoning about classical recursive programs. 

{\vskip 3pt}
 
\textbf{Organisation of This Paper}: For convenience of the reader, we briefly review Birkhoff-von Neumann (propositional and first-order) quantum logic in Section \ref{Sec-QL-P}. Our new first-order quantum logic $\mathcal{QL}$ with quantum variables is introduced in Section \ref{sec-QL-new}. As a preparation of the subsequent sections, Section \ref{sec-QP} is devoted to a brief review of the syntax and semantics of quantum while-programs. In Section \ref{Sec-QHL}, quantum Hoare logic (QHL) is recasted with $\mathcal{QH}$ defined in Section \ref{sec-QL-new} as its assertion logic. In particular, assertion logic $\mathcal{QL}$ enables us to formulate the relative completeness of QHL in a more formal way than \cite{Ying11,Zhou19}. As applications of a combined power of QHL and $\mathcal{QL}$, in Section \ref{sec-Applications}, we derive a series of adaptation rules for quantum programs, including a quantum generalisation of Hoare's adaptation rule which has played a crucial role in reasoning about procedure calls and recursion in classical programming \cite{Hoare71}, and show how they can be used to help runtime assertions \cite{Liu21}. The paper is concluded in Section \ref{sec-Con} where several open problems are pointed out. 

\section{Birkhoff-von Neumann Quantum Logic}\label{Sec-QL-P}

To set the stage, in this section we first recall some basic ideas of quantum logic (QL) from \cite{BvN36,Dalla04,Kal83}.
  
\subsection{Physical Observables and Propositions}
Let us start from seeing how we can describe a proposition about a quantum system. 
A piece of information about a physical system is called a \textit{state} of the system. In classical physics, a state of a system is usually described by a real vector, say an $n$-dimensional vector $\omega=(x_1,...,x_n)$ with all $x_i$ being real numbers, and the \textit{state space} $\Omega$ of the system is then the $n$-dimensional real vector space. A \textit{proposition} about a classical system asserts that a physical \textit{observable} has a certain value, and thus determines a subset $X\subseteq\Omega$ in which the proposition holds.
Thus, a state $\omega\in\Omega$ satisfies a proposition $X$, written $\omega\models X$, if and only if $\omega\in X$. 

According to the basic postulates of quantum mechanics, however, a \textit{state} of a quantum system is represented by a complex vector $|\psi\rangle$ (in Dirac's notation), and the \textit{state space} of the system is a Hilbert space $\hs$, i.e. a complex vector space equipped with an inner product $\langle\cdot |\cdot\rangle$ (satisfying certain completeness in the infinite-dimensional case). 
Then a \textit{proposition} asserting that a physical \textit{observable} has a certain value is mathematically represented by a closed subspace $X$ of $\hs$. 
A basic difference between a classical proposition and a quantum proposition is that the former can be an arbitrary subset of the state space, whereas the latter must be a closed subspace, that is, those subsets of $\hs$ closed under linear combination (and limit whenever $\hs$ is infinite-dimensional). 
Moreover, different from classical physics, quantum physical laws are essentially statistical. For any $|\psi\rangle\in\hs$ and $X\in\mathcal{P}$, the probability that the system in state $|\psi\rangle$ satisfies proposition $X$ is computed using Born's rule: $\textrm{Prob}(|\psi\rangle\models X)=\langle\psi|P_X|\psi\rangle,$ where $P_X$ is the projection onto closed subspace $X$. In particular, $\textrm{Prob}(|\psi\rangle\models X)=1$ if and only if $|\psi\rangle\in X$.      

\subsection{Operations of Closed Subspaces of a Hilbert Space}
For classical propositions, the logical connectives $\wedge$ (and), $\vee$ (or), $\neg$ (not) can be simply interpreted as set-theoretic operations $\cap,\cup$ and $^c$, respectively:
$$\omega\models X\wedge Y\ {\rm iff}\ \omega\in X\cap Y,\qquad \omega\models X\vee Y\ {\rm iff}\ \omega\in X\cup Y,\qquad \omega\models \neg X\ {\rm iff}\ \omega\in X^c.$$
However, this interpretation of connectives cannot be directly generalised to the quantum case because although $\cap$ preserves closeness under linear combination, $\cup$ and $^c$ do not.  
The operations appropriate for the interpretations of logical connectives $\wedge$ (and), $\vee$ (or), $\neg$ (not) for quantum propositions are defined as follows: for any closed subspaces $X,Y$ of $\hs$, 
\begin{itemize}
\item \textit{Meet}: $X\wedge Y=X\cap Y.$
\item \textit{Join}: $X\vee Y=\overline{\mathrm{span} (X\cup Y)}$, where for any subset $Z\subseteq\hs$, $\overline{\mathrm{span}\ Z}$ is the closed subspace of $\hs$ generated by $Z$; more precisely, $\overline{Z}$ stands for the topological closure of $Z$, and $\mathrm{span}\ Z$ is the smallest subspace of $\hs$ containing $Z$:
\begin{align*}\mathrm{span}\ Z&=\left\{\sum_{i=1}^n\alpha_i|\psi_i\rangle:n\geq 1, \alpha_i\in\mathbb{C}\ {\rm and}\ |\psi_i\rangle\in Z\ {\rm for\ all}\ 1\leq i\leq n \right\}.\end{align*}
\item\textit{Orthocomplement}: $X^\perp =\{|\psi\rangle\in\hs:|\psi\rangle\bot|\varphi\rangle\ {\rm for\ all}\ |\varphi\rangle\in X\},$ where $|\psi\rangle\bot |\varphi\rangle$ means that $|\psi\rangle$ and $|\varphi\rangle$ are orthogonal; that is, their inner product $\langle\psi|\varphi\rangle=0$.
\end{itemize}

For a classical system with state space $\Omega$, we write $2^\Omega$ for the power set of $\Omega$, i.e. the set of all subsets of $\Omega$. It is well known that $(2^\Omega,\cap, \cup,^c)$ is a (complete) Boolean algebra. Therefore, classical (Boolean) logic  is appropriate for reasoning about classical systems. Correspondingly, for a quantum system with state space $\hs$, let $S(\hs)$ stand for the set of all closed subspaces of $\hs$. Then:    

\begin{thm}[Sasaki 1954]\label{thm-Sasaki}\begin{enumerate}\item $(S(\hs),\wedge,\vee,\perp)$ is a complete orthomodular lattice, in which the partial order is set inclusion $\subseteq$, and the least and greatest elements are the $0$-dimensional closed subspace $\mathbf{0}=\{0\}$ and $\hs$, respectively. That is,   
$(\mathcal{S}(\hs),\wedge,\vee)$ is a complete lattice, and the following conditions are satisfied: for any $X,Y, Z\in S(\hs)$, \begin{itemize}
\item \textit{Ortho-modularity}: $X\subseteq Y$ implies $Y=X\vee (X^\perp\wedge Y).$
\item \textit{Contradiction and Excluded Middle Laws}: $X\wedge X^\perp =\mathbf{0}\ {\rm and}\ X\vee X^\perp=\hs.$   
\end{itemize} 
\item $(S(\hs),\wedge,\vee,\perp)$ is a modular lattice, that is, it satisfies the following:\begin{itemize}\item \textit{Modularity}: $X\subseteq Y$ implies $X\vee(Z\wedge Y)=(X\vee Z)\wedge Y$
\end{itemize} if and only if $\hs$ is finite-dimensional. 
\end{enumerate}\end{thm}

Obviously, modularity implies ortho-modularity. Furthermore, a Boolean algebra  satisfies:\begin{itemize}\item \textit{Distributivity}: $X\wedge (Y\vee Z)=(X\wedge Y)\vee (X\wedge Z)\ {\rm and}\  X\vee (Y\wedge Z)=(X\vee Y)\wedge (X\vee Z)$
\end{itemize} which is stronger than modularity.

\subsection{Propositional Quantum Logic}\label{subsec-pql}

As is well-known, classical logic is Boolean-valued, meaning that its truth values are taken from a Boolean algebra. Motivated by the above Sasaki theorem, quantum logic (QL) is defined as orthomodular lattice-valued logic. We first introduce propositional QL. It adopts a standard propositional language with an alphabet consisting of:\begin{itemize}\item[(i)] a set of propositional variables $P_0,P_1,P_2,...$; and \item[(ii)] connectives $\wedge$ (conjunction) and $\neg$ (negation).\end{itemize} Quantum propositional formulas can be defined in a familiar way. 
The disjunction is defined as a derived connective by $\beta\vee\gamma:=\neg(\neg\beta\wedge\neg\gamma).$ 

The semantics of propositional QL is then defined as follows.
Given an orthomodular lattice $\mathcal{L}=(L,\wedge,\vee,\perp)$. An $\mathcal{L}$-valued interpretation is a valuation function $v:\{P_0,P_1,P_2,...\}$ (propositional variables) $\rightarrow L,$ and it can be extended to all propositional formulas by the following valuation rules:
$$v(\beta\wedge \gamma)=v(\beta)\wedge v(\gamma),\qquad v(\neg \beta)=v(\beta)^\perp.$$ 
It can be verified that $v(\beta\vee \gamma)=v(\beta)\vee v(\gamma).$ Note that symbols $\wedge,\vee,\neg$ in the right-hand sides of these rules denote the operations in $\mathcal{L}$.     
For a set $\Sigma$ of propositional formulas and a proposition formula $\beta$, $\beta$ is called a consequence of $\Sigma$ in $\mathcal{L}$, written $\Sigma\models_\mathcal{L} \beta$ if for any valuation function $v$, and for any $a\in L$: whenever $a\leq v(\gamma)$ for all $\gamma\in\Sigma$, then $a\leq v(\beta).$ When $\Sigma$ is finite, say $\Sigma=\{\gamma_1,...,\gamma_n\}$, then $\Sigma\models_\mathcal{L} \beta$ if and only in for any valuation $v$: $\bigwedge_{i=1}^nv(\gamma_i)\leq v(\beta).$ In particular, $\emptyset\models_\mathcal{L} \beta$ if and only if $v(\beta)=\mathbf{1}$ for all valuation function $v$, where $\mathbf{1}$ is the greatest element of $\mathcal{L}$. In this case, $\beta$ is said to be true in $\mathcal{L}$ and we write $\models_\mathcal{L} \beta$. 
 
An axiomatisation of propositional QL in the Gentzen-style is presented in Figure \ref{fig QL-P}. 
\begin{figure}[h]\centering
\begin{equation*}\begin{split}
&({\rm QL1})\ \ \ \Sigma\cup\{\beta\}\vdash \beta\qquad\qquad\qquad\qquad\qquad\quad\quad\ ({\rm QL2})\ \ \ \frac{\Sigma\vdash \beta\qquad \Sigma^\prime\cup\{\beta\}\vdash \gamma}{\Sigma\cup\Sigma^\prime\vdash \gamma}\\
&({\rm QL3})\ \ \ \Sigma\cup\{\beta\wedge \gamma\}\vdash \beta\quad \Sigma\cup\{\beta\wedge \gamma\}\vdash \gamma\quad\quad\ \ ({\rm QL4})\ \ \ \frac{\Sigma\vdash \beta\qquad\Sigma\vdash \gamma}{\Sigma\vdash \beta\wedge \gamma}\\
&({\rm QL5})\ \ \ \frac{\Sigma\cup\{\beta,\gamma\}\vdash \delta}{\Sigma\cup\{\beta\wedge \gamma\}\vdash \delta}\qquad\qquad\qquad\qquad\quad\ \ \ \ \ ({\rm QL6})\ \ \ \frac{\beta\vdash \gamma\qquad \beta\vdash\neg \gamma}{\neg \beta}\\
&({\rm QL7})\ \ \ \Sigma\cup\{\beta\}\vdash\neg\neg \beta\qquad\qquad\qquad\qquad\qquad\quad\ ({\rm QL8})\ \ \ \Sigma\cup\{\neg\neg \beta\}\vdash \beta\\
&({\rm QL9})\ \ \ \Sigma\cup\{\beta\wedge\neg \beta\}\vdash \gamma\qquad\qquad\qquad\qquad\qquad\ ({\rm QL10})\ \ \ \frac{\beta\vdash \gamma}{\neg \gamma\vdash \neg \beta}\\
&({\rm QL11})\ \ \ \beta\wedge \neg(\beta\wedge\neg (\beta\wedge \gamma))\vdash \gamma
\end{split}\end{equation*}
\caption{Axiomatic System of Propositional QL.}\label{fig QL-P}
\end{figure}

\begin{rem}\label{remark-Sasaki}
All implications that can be reasonably defined in QL are anomalous to a certain extent. A minimal requirement for an operation $\rightarrow$ in an orthomodular lattice $\mathcal{L}$ that can serve as an interpretation of implication is that for all $a,b\in \mathcal{L}$, $a\rightarrow b= 1$ iff $a\leq b$. It was proved that there are only five such operations that can be defined in terms of $\wedge,\vee,\perp$. Among them, only the Sasaki implication:
$$a\rightarrow b=a^\perp \vee (a\wedge b)$$    satisfies the import-export condition: for all $a,b,c\in\mathcal{L}$, $a\wedge b\leq c$ iff $a\leq b\rightarrow c$. In this paper, we always use the Sasaki implication. 
\end{rem}

\subsection{First-Order Quantum Logic with Classical Variables}\label{Sec-QL-FO}

Now let us move on to consider first-order QL. It uses a standard first-order language with the following alphabet: \begin{itemize}\item[(i)] a set of individual variables $x,y,z,...$; \item[(ii)] a set of function symbols $f,g,...$ (including constants $c,d,...$ as $0$-ary functions); \item[(iii)] a set of predicate symbols $P,Q,,...$; \item[(iv)] connectives $\wedge, \neg$; and \item[(v)] universal quantifier $\forall$. \end{itemize}
The terms and first-order logical formulas in QL are also defined in a familiar way. 
The existential quantification can be defined as a derived formula: $(\exists x)\gamma:=\neg (\forall x)\neg \gamma$. 

Let $\mathcal{L}=(L,\wedge,\vee,\perp)$ be a complete orthomodular lattice. Then an $\mathcal{L}$-valued interpretation $\mathbb{I}$ of logic QL consists of:\begin{itemize}\item a nonempty set $D$, called the domain of $\mathbb{I}$; 
\item for each $n$-ary function symbol $f$, it is interpreted as a mapping $f^\mathbb{I}:D^n\rightarrow D$. In particular, a constant $c$ is interpreted as an element $c^\mathbb{I}\in D$; \item for each $n$-ary predicate symbol $P$, it is interpreted as an $\mathcal{L}$-valued relation, i.e., a mapping $P^\mathbb{I}:D^n\rightarrow L$.  
\end{itemize} Given an interpretation $\mathbb{I}$ and a valuation function $\sigma:\{x,y,z,...\}\ ({\rm individual\ variables}) \rightarrow D$. They 
define the semantics of terms exactly in the same way as in classical first order logic. Furthermore, 
they define a truth valuation:
\begin{itemize}\item for an atomic formula $\beta=P(t_1,...,t_n)$, $v(\beta)=P^\mathbb{I}(v(t_1),...,v(t_n))$, where $v(t_i)$ stands for the value of term $t_i$;
\item $v(\beta\wedge\gamma)=v(\beta)\wedge v(\gamma)$ and $v(\neg \beta)=v(\beta)^\perp$;
\item $v((\forall x)\beta)=\bigwedge\{v[d/x](\beta):d\in D\}$, where $v[d/x]$ is the truth valuation defined by the same model $\mathcal{M}$ together with valuation function $\sigma[d/x]$, which coincides with $\sigma$ except that $\sigma[d/x](x)=d$. 
\end{itemize} 
The valuation rule for existential quantifier can be derived as $v((\exists x)\beta)=\bigvee\{v[d/x](\beta):d\in D\}$. 

It should be particularly noted that in QL individual variables $x,y,z,...$ are still classical variables with values taken from an ordinary domain $D$. The only difference between QL and classical first order logic is that the set of truth values in QL is an orthomodular lattice $\mathcal{L}$, which is set to be the lattice $S(\hs)$ of closed subspaces of $\hs$ when QL is applied to specify and reason about a quantum system with state space $\hs$.     

An axiomatisation of first-order QL can be obtained by adding the two rules for universal quantifier presented in Figure \ref{fig QL-FO} to the axiomatic system of propositional QL given in Figure \ref{fig QL-P}.
\begin{figure}[h]\centering
\begin{equation*}\begin{split}
&({\rm QL12})\ \ \ \Sigma\cup\{(\forall x)\beta\}\vdash\beta[t/x]\qquad\qquad\qquad({\rm QL13})\ \ \ \frac{\Sigma\vdash\beta}{\Sigma\vdash (\forall x)\beta}\ (x\ {\rm is\ not\ free\ in}\ \Sigma)
\end{split}\end{equation*}
\caption{Axiomatic System of First-order QL. In (QL12), $\beta[t/x]$ stands for the substitution of variable $x$ by term $t$ in $\beta$.}\label{fig QL-FO}
\end{figure}

It was proved that the axiomatic system of (propositional and first-order) QL given in Figures \ref{fig QL-P} and \ref{fig QL-FO} is complete with respect to orthomular lattice-valued semantics; that is, we have:

\begin{thm}[Goldblatt 1974]\label{thm-QL-complete} $\Sigma\vdash P$ iff $\Sigma\models_\mathcal{L} P$ for every orthomodular lattice $\mathcal{L}.$
\end{thm}

\section{First-Order Quantum Logic with Quantum Variables}\label{sec-QL-new}

As pointed out in the Introduction, QL with classical variables presented in the previous section is often inconvenient when used in quantum computation and information. In this section, we turn to define our new logic --- a first-order quantum logic $\mathcal{QL}$ with quantum variables.  

\subsection{Syntax of $\mathcal{QL}$}

The syntax of $\mathcal{QL}$ is similar to that of QL except that classical individual variables in QL are replaced by quantum variables in $\mathcal{QL}$. Accordingly, 
function symbols in QL are replaced by symbols that denote quantum operations over (the states of) quantum variables. The predicate symbols in both QL and $\mathcal{QL}$ are interpreted as an orthomodular lattice-valued functions, but in QL they are over a domain of classical individuals and in $\mathcal{QL}$ they are over a domain of quantum states.    

\subsubsection{Alphabet} Formally, the alphabet of $\mathcal{QL}$ consists of:\begin{enumerate}
\item A set $\mathit{Var}$ of quantum variables $q,q_1,q_2,...$;
\item A set of quantum operation symbols $\mathcal{E},\mathcal{E}_1,\mathcal{E}_2,...$, with a subset of unitary symbols $\mathcal{U},\mathcal{U}_1,\mathcal{U}_2,...$ and their inverses $\mathcal{U}^{-1},\mathcal{U}^{-1}_1,\mathcal{U}^{-1}_2,...$; 
\item A set of quantum predicate symbols $P,Q,...$;
\item Connectives $\neg,\wedge$;
\item Universal quantifier $\forall$. 
\end{enumerate}
To each quantum variable $q\in\mathit{Var}$, a nonnegative integer $d$ or $d=\infty$ is assigned, called the dimension of $q$. To each quantum operation symbol $\mathcal{E}$, a nonnegative integer $n$ and an $n$-tuple $\overline{d}=(d_1,..., d_n)$ of nonnegative integers or $\infty$ are assigned, called the arity and signature of $P$, respectively. To each quantum predicate symbol $P$, an arity and a signature are assigned too.   

\subsubsection{Terms} The notions of term in classical first-order logic can be straightforwardly generalised into $\mathcal{QL}$, with some formation rules specifically for modelling quantum operations.  
\begin{defn}\label{def-terms}Quantum terms $\tau$ and their variables $\mathit{var}(\tau)$ are inductively defined as follows:\begin{enumerate}
\item (Basic terms) If $\mathcal{E}$ is an $n$-ary quantum operation symbol with signature $(d_1,...,d_n)$, and $\overline{q}=q_1,...,q_n$ where $q_i$ is a $d_i$-dimensional quantum variable for each $1\leq i\leq n$, then $\tau=\mathcal{E}(\overline{q})$ is a quantum term and $\mathit{var}(\tau)=\overline{q}$;   
\item (Sequential composition) If $\tau_1,\tau_2$ are quantum terms, so is $\tau=\tau_1\tau_2$ and $\mathit{var}(\tau)=\mathit{var}(\tau_1)\cup\mathit{var}(\tau_2)$; 
\item (Tensor product) If $\tau_1,\tau_2$ are quantum terms and $\mathit{var}(\tau_1)\cap\mathit{var}(\tau_2)=\emptyset$, then $\tau=\tau_1\otimes\tau_2$ is a quantum term and $\mathit{var}(\tau)=\mathit{var}(\tau_1)\cup\mathit{var}(\tau_2)$;
\item (Probabilistic combination) If $\{\tau_i\}$ is a family of quantum terms with the same variables $\mathit{var}(\tau_i)=V$, and $\{p_i\}$ is a sub-probability distribution; that is, $p_i> 0$ for all $i$ and $\sum_ip_i\leq 1$, then $\tau=\sum_ip_i\tau_i$ is a quantum term, and $\mathit{var}(\tau)=V$.   
\end{enumerate}
In particular, if $\tau$ is generated only by clauses (1) - (3) and all quantum operation symbols in $\tau$ are unitary symbols, then $\tau$ is called a unitary term, and its inverse $\tau^{-1}$ is defined as follows:\begin{enumerate}
\item[(i)] If $\tau=\mathcal{U}(\overline{q})$, then $\tau^{-1}=\mathcal{U}^{-1}(\overline{q})$;
\item[(ii)] If $\tau=\tau_1\otimes\tau_2$, then $\tau^{-1}=\tau_1^{-1}\otimes\tau_2^{-1}$;
\item[(iii)] If $\tau=\tau_1\tau_2$, then $\tau^{-1}=\tau_2^{-1}\tau_1^{-1}$. 
\end{enumerate}
\end{defn} 

Clause (1) in the above definition defines the basic quantum operations. Clauses (2) and (3) are introduced for describing the sequential composition of two quantum operations and a separable operation on a composed system, respectively. 
The term $\tau_1\otimes\tau_2$ can also be understood as the parallel composition of $\tau_1$ and $\tau_2$. 
Quantum terms defined by clause (4) are introduced for modelling a probabilistic combination of quantum states, in particular for merging the outcomes from different branches of a computation; i.e. a mixed state that is formed as an ensemble of output quantum states from different branches of the computation. It is interesting to note that tensor product $\tau_1\otimes\tau_2$ is equivalent to sequential composition $\tau_1\tau_2$ because it is required in its definition that $\mathit{var}(\tau_1)\cap \mathit{var}(\tau_2)=\emptyset.$ However, a probabilistic combination $\sum_ip_i(\tau_{1i}\otimes\tau_{2i})$ of multiple tensor products cannot always be expressed by sequential composition.   

\begin{exam}\label{noisy-term}\label{example-noisy}The term $\tau$ defined by equation (\ref{term-exam}) in Example \ref{example-1} is a unitary term that expresses the quantum circuit in Figure \ref{combinational-quantum-circuit}. Obviously, all (combinational) quantum circuits, including noisy quantum circuits, can be written as quantum terms. For example, 
if a bit-flip noise $\E_{bf}$ occurs immediately after the gate $Z$ and a phase-flip noise $\E_{pf}$ occurs on qubit $q_2$ after the CNOT gate n Figure \ref{combinational-quantum-circuit}, then the noisy circuit can be  
 written as the following term:
\begin{equation}\label{term-exam2}\begin{split}\tau^\prime&=Z(q_1)\E_{bf}(q_1)H(q_2)C(q_1,q_2)\E_{pf}(q_2)Y(q_1)H(q_2)\\
&=[Z(q_1)\E_{bf}(q_1)\otimes H(q_2)]C(q_1,q_2)[Y(q_1)\otimes \E_{pf}(q_2)H(q_2)].\end{split}\end{equation}\end{exam}

\subsubsection{Logical Formulas}\label{sec-logic-formula}

The logical formulas in $\mathcal{QL}$ are also straightforward generalisation of the standard first-order logical formulas excepted those given by clause (4) in the following definition. 

\begin{defn}\label{def-logic-formula}The formulas $\beta$ of logic $\mathcal{QL}$ and their free variables $\mathit{free}(\beta)$ are inductively defined as follows:
\begin{enumerate}
\item If $P$ is an $n$-ary quantum predicate symbol with signature $(d_1,...,d_n)$, and $\tau$ a term with $\mathit{var}(\tau)=\overline{q}=q_1,...,q_n$, where $q_i$ is a $d_i$-dimensional quantum variable for each $1\leq i\leq n$, then $\beta=P(\tau)$ is a formula and $\mathit{free}(\beta)=\mathit{var}(\tau)$; 
\item If $\beta^\prime$ is a formula, so is $\beta=\neg\beta^\prime$ and $\mathit{free}(\beta)=\mathit{free}(\beta^\prime)$; \item If $\beta_1,\beta_2$ are formulas, so is $\beta=\beta_1\wedge\beta_2$ and $\mathit{free}(\beta)=\mathit{free}(\beta_1)\cup\mathit{free}(\beta_2)$;
\item If $\beta^\prime$ is a formula and $\tau$ a term, then $\beta=\tau^\ast(\beta^\prime)$ is a formula and $\mathit{free}(\beta)=\mathit{var}(\tau)\cup\mathit{free}(\beta^\prime)$; 
\item If $\beta^\prime$ is a formula and $\overline{q}$ is a sequence of quantum variables, then $\beta=(\forall \overline{q})\beta^\prime$ is a formula and $\mathit{free}(\beta)=\mathit{free}(\beta^\prime)\setminus\overline{q}$;
\end{enumerate} 
\end{defn}

If $\mathcal{I}$ is a symbol for the identity operation, then $P(\mathcal{I}(q_1)...\mathcal{I}(q_n))$ is a formula, often written as $P(q_1,...,q_n)$ for simplicity. The existential quantification can be defined as a derived formula:
$(\exists\overline{q})\beta=\neg (\forall\overline{q})\neg\beta.$

\begin{exam}\label{exam-wff} We use the quantum variables and quantum operation symbols in Example \ref{example-noisy}. Moreover, let $P_0,P_1$ be two quantum predicate symbols for a single qubit, and $P_e$ a quantum predicate symbol for two qubits. The the following are two logical formulas in $\mathcal{QL}$:\begin{enumerate}\item $\beta_1 = \neg P_0(Z(q_1)\mathcal{E}_\mathit{bf}(q_1))\wedge P_2(H(q_2))$;
\item $\beta_2 = P_e(\tau^\prime)\wedge (\forall q_2)P_2(q_2)$, where $\tau^\prime$ is the quantum term given in equation (\ref{term-exam2}).
\end{enumerate} Intuitively, formula $\beta_1$ expresses that after the Pauli gate and bit-flip noise, the state of qubit $q_1$ is in the subspace denoted by $P_0$, and after the Hadamard gate, the state of $q_2$ is in the subspace denoted by $P_2$; formula $\beta_2$ says that the output of the noisy circuit (\ref{term-exam2}) is in the subspace denoted by $P_e$, and after any allowed operation, the state of $q_2$ is still in the subspace denoted by $P_2$.   
\end{exam}

Clause (4) in the above definition shows a fundamental difference between classical first-order logic and our logic $\mathcal{QL}$ and deserve a careful explanation. 
A formula of the form $\beta=\tau^\ast(\beta^\prime)$ is called a \textit{term-adjoint formula}. Essentially, whenever $\tau$ contains a single quantum variable $q$, i.e. $\mathit{var}(\tau)=q$, it is the quantum version of substitution $\beta^\prime[t/x]$ of variable $x$ in logical formula $\beta^\prime$ by term $t$. In classical logic, $\beta^\prime[t/x]$ is obtained by substituting all free occurrences of $x$ in $\beta^\prime$ with $t$. However, in the quantum case, substitution cannot be defined in such a way due to the so-called Schr\"{o}dinger-Heisenberg duality. Whence an interpretation is given, a quantum term $\tau$ denotes a quantum state, which should be considered in the Schr\"{o}dinger picture. On the other hand, a $\mathcal{QL}$ formula $\beta^\prime$ denotes a closed subspace of the state Hilbert space (equivalently, a projection operator as a special form of observable) and thus should be considered in the Heisenberg picture. Consequently, when applying term $\tau$ to modify formula $\beta^\prime$, we must use the dual $\tau^\ast$ of $\tau$ rather than $\tau$ itself. If $\tau$ contains more than one quantum variable, say $\mathit{var}(\tau)=q_1...q_n$, then $\tau^\ast(\beta^\prime)$ can be understood as a quantum analogue of simultaneous substitution $\beta^\prime[t_1/x_1,...,t_n.x_n]$, but we must keep in mind that $\tau$ may denote an entangled state of $q_1,...,q_n$, and thus $q_1,...,q_n$ cannot be separately substituted. This point will be seen more clearly from the semantics of $\beta=\tau^\ast(\beta^\prime)$ below. As we will see in Subsection \ref{sec-QHL-axioms}, term-adjoint formulas are needed in defining the proof rules for some basic quantum programs. 

\subsection{Semantics of $\mathcal{QL}$}\label{subsec-ql-semantics}

In this subsection, we define the semantics of $\mathcal{QL}$ formulas, from which we will see various differences between $\mathcal{QL}$ and QL. 

\subsubsection{Interpretations} 

First of all, each individual variable $q$ in $\mathcal{QL}$ is a quantum variable, and its values are quantum states in its state space $\hs_q$. For a $d$-dimensional quantum variable $q$, if $d<\infty$ then its state space is (isomorphic to) the Hilbert space with orthonormal basis $\{|0\rangle,...,|d-1\rangle\}$: $$\hs_q=\left\{\sum_{i=0}^{d-1}c_i|i\rangle:c_i\in\mathbb{C}\ (0\leq i<d)\right\}.$$ 
In this paper, we only consider separable Hilbert spaces. Thus, if $d=\infty$ then we can assume: \begin{equation}\label{infinite-Hilbert}\hs_q=\left\{\sum_{i=-\infty}^{\infty}c_i|i\rangle:c_i\in\mathbb{C}\ (0\leq i<d)\ {\rm with}\ \sum_{i=-\infty}^\infty|c_i|^2<\infty\right\}\end{equation} with orthonormal basis $\{|i\rangle:i\in\mathbb{Z}\ {\rm (intergers)}\}$. For any subset of quantum variables $V\subseteq\mathit{Var}$, we write $\hs_V=\bigotimes_{q\in V}\hs_q$, and the identity quantum operation on $\hs_V$ is denoted $\mathcal{I}_V$. In particular, $\hs_\mathit{Var}$ is the state space of all quantum variables assumed in our logic. 
Secondly, we recall from \cite{NC00} that a quantum operation on a Hilbert space $\hs$ is defined as a completely positive and trace-non-increasing super-operator, i.e. a linear map from operators on $\hs$ to themselves.       
We assume a family $\mathcal{O}=\left\{\mathcal{O}_{\overline{d}}\right\}$, where for each signature $\overline{d}=(d_1,...,d_n)$, $\mathcal{O}_{\overline{d}}$ is a set of quantum operations on the $\prod_{i=1}^nd_i$-dimensional Hilbert space, called \textit{allowed operations}. Thirdly, recall from the last section that the first-order QL in Section \ref{Sec-QL-FO} is interpreted in an arbitrary orthomodular lattice $\mathcal{L}$. However, our logic $\mathcal{QL}$ is only interpreted in a special class of $\mathcal{L}=S(\hs)$ (the orthomodular lattice of closed  subspaces of $\hs$) for certain Hilbert spaces $\hs$. Formally, we have: 
\begin{defn}An interpretation $\mathbb{I}$ of logic $\mathcal{QL}$ is defined as follows:\begin{enumerate}\item To each $d$-dimensional quantum variable $q\in\mathit{Var}$, a $d$-dimensional Hilbert space $\hs_q$ is associated, called the state space of $q$; 
\item Each $n$-ary quantum operation symbol $\mathcal{E}$ with signature $(d_1,...,d_n)$ is interpreted as an allowed quantum operation $\mathcal{E}^\mathbb{I}\in\mathcal{O}_{\overline{d}}$ on the $\prod_{i=1}^nd_i$-dimensional Hilbert space. It is required that for each unitary symbol $\mathcal{U}$, $(\mathcal{U}^{-1})^\mathbb{I}=(\mathcal{U}^\mathbb{I})^{-1}$; 
\item Each $n$-ary quantum predicate symbol $\mathcal{E}$ with signature $(d_1,...,d_n)$ is interpreted as a closed subspace $P^\mathbb{I}$ of the $\prod_{i=1}^nd_i$-dimensional Hilbert space.
\end{enumerate}\end{defn}

The difference between an interpretation of $\mathcal{QL}$ and that of QL defined in Section \ref{Sec-QL-FO} is obvious. Essentially, all other differences between $\mathcal{QL}$ and QL originate from it. Furthermore, the (meta-)logical properties of $\mathcal{QL}$ heavily depends on the allowed quantum operations $\mathcal{O}=\left\{\mathcal{O}_{\overline{d}}\right\}$.  

\subsubsection{Semantics of Terms} 

Now let us define the semantics of quantum terms. As mentioned in the last section, a quantum term can be considered in two different ways. In the Schr\"{o}dinger picture, it is interpreted as a mapping from quantum states to quantum states. Recall from \cite{NC00} that a (mixed) state of a quantum system with Hilbert space $\hs$ as its state space is described as a density operator on $\hs$.  
An operator $\rho$ on $\hs$ is called a partial density operator if it is positive and $\tr(\rho)\leq 1$. In particular, if $\tr(\rho)= 1$, then $\rho$ is called a density operator.   
We use $\mathcal{D}(\hs)$ to denote the set of partial density operators on $\hs$. 

\begin{defn}\label{def-sem-terms} Given an interpretation $\mathbb{I}$. The Schr\"{o}dinger semantics of a term $\tau$ is a mapping $\llbracket \tau\rrbracket_\mathbb{I}:\mathcal{D}(\mathcal{H}_{\mathit{Var}})\rightarrow \mathcal{D}(\mathcal{H}_{\mathit{Var}}).$ For each $\rho\in \mathcal{D}(\mathcal{H}_{\mathit{Var}})$, $\llbracket \tau\rrbracket_\mathbb{I}(\rho)$ is defined as follows: \begin{enumerate}\item If $\tau=\mathcal{E}(\overline{q})$, then 
$\llbracket \tau\rrbracket_\mathbb{I}(\rho)=\left(\mathcal{E}^\mathbb{I}\otimes \mathcal{I}_{\mathit{Var}\setminus\overline{q}}\right)(\rho);$ 
\item If $\tau=\tau_1\tau_2$, then $\llbracket \tau\rrbracket_\mathbb{I}(\rho)=\llbracket \tau_2\rrbracket_\mathbb{I}\left(\llbracket \tau_1\rrbracket_\mathbb{I}(\rho)\right);$
\item If $\tau=\tau_1\otimes\tau_2$, then $\llbracket \tau\rrbracket_\mathbb{I}(\rho)=\llbracket \tau_2\rrbracket_\mathbb{I}\left(\llbracket \tau_1\rrbracket_\mathbb{I}(\rho)\right)=\llbracket \tau_1\rrbracket_\mathbb{I}\left(\llbracket \tau_2\rrbracket_\mathbb{I}(\rho)\right);$
\item If $\tau=\sum_ip_i\tau_i$, then $\llbracket\tau\rrbracket_\mathbb{I}(\rho)=\sum_ip_i\llbracket\tau_i\rrbracket_\mathbb{I}(\rho).$
\end{enumerate}\end{defn}

It is easy to show that $\llbracket \tau\rrbracket_\mathbb{I}$ is a quantum operation (i.e. completely positive super-operator that does not increase trace) on $\hs_{\mathit{Var}}$. 

In the Heisenberg picture, however, a quantum term should be interpreted as a mapping from observables to observables. In this paper, we focus on so-called sharp quantum logic with closed subspaces (equivalently, projection operators) as logical propositions, and thus the Heisenberg interpretation of a term is defined as a mapping from subspaces to themselves. Let $\mathcal{E}$ be a quantum operation on Hilbert space $\hs$ and $X\in\mathcal{S}(\hs)$. Then the image of $X$ under $\E$ is defined as  
\begin{equation}\label{definition-image}\E(X)=\bigvee_{|\psi\rangle\in X}\supp[\E(|\psi\rangle\langle\psi)]\end{equation} where $\supp(\rho)$ denotes the support of $\rho$, i.e. the subspace spanned by the eigenvectors of $\rho$ corresponding to nonzero eigenvalues.

\begin{defn} Given an interpretation $\mathbb{I}$. The Heisenberg semantics of a term $\tau$ is a mapping $\llbracket \tau\rrbracket_\mathbb{I}^\ast:S(\mathcal{H}_{\mathit{Var}})\rightarrow S(\mathcal{H}_{\mathit{Var}})$. For each $X\in S(\mathcal{H}_{\mathit{Var}})$, $\llbracket \tau\rrbracket_\mathbb{I}^\ast(X)$ is defined as follows: \begin{enumerate}\item If $\tau=\mathcal{E}(\overline{q})$, then 
$\llbracket \tau\rrbracket_\mathbb{I}^\ast(X)=\left((\mathcal{E}^\mathbb{I})^\ast\otimes \mathcal{I}_{\mathit{Var}\setminus\overline{q}}\right)(X);$ 
\item If $\tau=\tau_1\tau_2$, then $\llbracket \tau\rrbracket_\mathbb{I}^\ast(X)=\llbracket \tau_1\rrbracket_\mathbb{I}^\ast\left(\llbracket \tau_2\rrbracket_\mathbb{I}^\ast(X)\right);$
\item If $\tau=\tau_1\otimes\tau_2$, then $\llbracket \tau\rrbracket_\mathbb{I}^\ast(X)=\llbracket \tau_1\rrbracket_\mathbb{I}^\ast\left(\llbracket \tau_2\rrbracket_\mathbb{I}^\ast(X)\right)=\llbracket \tau_2\rrbracket_\mathbb{I}^\ast\left(\llbracket \tau_1\rrbracket_\mathbb{I}^\ast(X)\right);$
\item If $\tau=\sum_ip_i\tau_i$, then $\llbracket\tau\rrbracket_\mathbb{I}^\ast(X)=\bigvee_i \llbracket\tau_i\rrbracket_\mathbb{I}^\ast(X).$
\end{enumerate}\end{defn}

It should be noted that as usual for simplicity of presentation, both the Schr\"{o}dinger and Heisenberg semantics of a term $\tau$ are defined on the state space $\hs_\mathit{Var}$ of all variables $\mathit{Var}$. But we can show that only the variables $\mathit{var}(\tau)$ appearing in $\tau$ are essential for them. To this end, let us recall that for any density operator $\rho\in\mathcal{D}(\hs_1\otimes\hs_2)$, its restriction on $\hs_1$ is $\rho\downarrow \hs_1=\tr_{\hs_2}(\rho)$, where partial trace $\tr_{\hs_2}$ over $\hs_2$ is defined by $\tr_{\hs_2}(|\varphi_1\rangle\langle\psi_1|\otimes |\varphi_2\rangle\langle\psi_2|)=\langle\psi_2|\varphi_2\rangle\cdot |\varphi_1\rangle\langle\psi_1|$ for any $|\varphi_i\rangle, |\psi_i\rangle\in\hs_i$ $(i=1,2),$ together with linearity. In particular, if $\rho\in\mathcal{D}(\hs_V)$ for some $V\subseteq\mathit{Var}$ and $V^\prime\subseteq V$, we simply write $\rho\downarrow V^\prime$ for $\rho\downarrow\hs_{V^\prime}.$ On the other hand, for any subspace $X\in\mathcal{S}(\hs_1\otimes\hs_2)$, its restriction on $\hs_1$ is defined as $X\downarrow \hs_1=\supp[\tr_{\hs_2}(P_X)]$, 
where $P_X$ is the projection operator onto $X$. In particular, if $X\in S(\hs_V)$ for some $V\subseteq\mathit{Var}$ and $V^\prime\subseteq V$, we simply write $X\downarrow V^\prime$ for $X\downarrow\hs_{V^\prime}.$ Then we have: 

\begin{lem}[Coincidence]\label{term-coincide} Let $\mathit{var}(\tau)\subseteq V\subseteq\mathit{Var}$. Then: \begin{enumerate}\item $\rho\downarrow V=\rho^\prime\downarrow V$ implies $\llbracket \tau\rrbracket_\mathbb{I}(\rho)\downarrow V=\llbracket \tau\rrbracket_\mathbb{I}(\rho^\prime)\downarrow V;$
\item $X\downarrow V=X^\prime\downarrow V$ implies $\llbracket \tau\rrbracket_\mathbb{I}^\ast(X)\downarrow V=\llbracket \tau\rrbracket_\mathbb{I}^\ast(X^\prime)\downarrow V.$ 
\end{enumerate}
\end{lem}

We already mentioned that certain duality exists between the Schr\"{o}dinger and Heisenberg semantics of quantum terms. It is precisely described in the following:   

\begin{lem}[Schr\"{o}dinger-Heisenberg Duality]\label{lem-duality}
	For any term $\tau$, density operator $\rho$, and closed subspace $X$, we have: $\sem{\tau}(\rho)\in X\Leftrightarrow\rho\in \left(\semp{\tau}(X^\perp)\right)^\perp.$ 
	In particular, if $\tau$ is unitary, or the allowed quantum operations in $\mathcal{O}_{\overline{d}}$ for the signature $\overline{d}$ of all operation symbols appearing in $\tau$ are unitary, then
	$\sem{\tau}(\rho)\in X\Leftrightarrow\rho\in \semp{\tau}(X).$
\end{lem}

\subsubsection{Semantics of Logical Formulas}\label{ss-section-sem}

We now move on to define the semantics of logical formulas in $\mathcal{QL}$. To this end, we need the following notations: 
\begin{itemize}\item Let $\rho\in\mathcal{D}(\hs)$ and $X\in S(\hs)$. Then we define: $\rho\in X\ {\rm iff}\ \supp(\rho)\subseteq X.$
\item Two partial density operators $\rho_1,\rho_2\in \mathcal{D}(\hs)$ are orthogonal, written $\rho_1\bot\rho_2$, if $\supp(\rho_1)\bot\supp(\rho_2)$.  
\end{itemize}

\begin{defn}\label{QFO-semantics}Given an interpretation $\mathbb{I}$ and a state $\rho\in\mathcal{D}(\hs_\mathit{Var})$. Let $\beta$ be a formula. Then satisfaction relation $(\mathbb{I},\rho)\models\beta$ is inductively defined as follows:
\begin{enumerate}
\item If $\beta=P(\tau)$, then $(\mathbb{I},\rho)\models\beta$ iff $\llbracket \tau\rrbracket_\mathbb{I}(\rho)\downarrow\mathit{var}(\tau)\in P^\mathbb{I}$;
\item If $\beta=\neg\beta^\prime$, then $(\mathbb{I},\rho)\models\beta$ iff for all $\rho^\prime$, $(\mathbb{I},\rho^\prime)\models\beta^\prime$ implies $\rho^\prime\bot\rho$;  
\item If $\beta=\beta_1\wedge\beta_2$, then $(\mathbb{I},\rho)\models\beta$ iff $(\mathbb{I},\rho)\models\beta_1$ and $(\mathbb{I},\rho)\models\beta_2$;
\item If $\beta=\tau^\ast(\beta^\prime)$, then $(\mathbb{I},\rho)\models\beta$ iff $(\mathbb{I},\llbracket \tau\rrbracket_\mathbb{I}(\rho))\models\beta^\prime;$
\item If $\beta=(\forall\overline{q})\beta^\prime$, then $(\mathbb{I},\rho)\models\beta$ iff for any term $\tau$ with $\mathit{var}(\tau)\subseteq\overline{q}$, it holds that $(\mathbb{I},\rho)\models\tau^\ast(\beta^\prime)$.
\end{enumerate}
\end{defn}

Clauses (1) and (3) are easy to understand. From clause (2) we see that if $(\mathbb{I},\rho)\models\neg\beta^\prime$ then $(\mathbb{I},\rho)\models\beta^\prime$ does not hold, but not vice versa. Indeed, $(\mathbb{I},\rho)\models\neg\beta^\prime$ means that   $(\mathbb{I},\rho^\prime)\models\beta^\prime$ is not true for all $\rho^\prime$ that is not orthogonal to $\rho$. Clause (4) reflects the Schr\"{o}dinger-Heisenberg duality discussed before. Clause (5) is essentially a restatement of equation (\ref{quantum-quantifier}), but here the allowed operations $\mathcal{E}\in\mathcal{O}_{\overline{q}}$ in equation (\ref{quantum-quantifier}) are syntactically expressed by quantum terms $\tau$. In classical program logics (e.g. Hoare logic and separation logic), substitution is needed in defining the proof rules for some basic program constructs. We pointed out in Subsection \ref{sec-logic-formula} that a term-adjoint formula of the form $\tau^\ast(\beta)$ is introduced for a role in our $\mathcal{QL}$ as the one of substitution in classical first-order logic. But its definition (see clause (4) of Definition \ref{def-logic-formula}) is quite different from that in the classical case: the former is defined as a primitive syntactic notion, whereas the latter is defined as a derived syntactic notion. A discussion about the relationship between semantics of term-adjoint formulas and substitutions is given in the Appendix.

To illustrate the above definition, let us see two examples. The first one is a continuation of Example \ref{example-1}: 

\begin{exam} Let $\mathbb{I}$ be the usual interpretation where $H, X, Y, Z, C$ denotes the Hadamard gate, Pauli gates and CNOT, respectively.
We consider the logical formula in Example \ref{example-1}: $$\beta=(\forall q_1)(\forall q_2)\beta^\prime,\ {\rm where}\ \beta^\prime=P_0(q_1)\wedge P(q_1,q_2)\rightarrow P(\tau).$$ It is easy to show that $(\mathbb{I},|0\rangle|0\rangle)\models\beta$. Intuitively, the two-qubit system $q_1q_2$ is initialised in basis state $|0\rangle|0\rangle$. For $i=1,2,$ each allowed operation on $q_i$ can be syntactically expressed by a quantum term $\tau_i$. After $\tau_i$, the subsystem $q_i$ is prepared in state $\rho_i=\llbracket\tau_i\rrbracket_\mathbb{I}(|0\rangle).$ Then the universal quantifications $(\forall q_1), (\forall q_2)$ mean that for any preparation operations $\tau_1,\tau_2$, the product state $\rho_1\otimes\rho_2$ satisfies $\beta^\prime$. Indeed, we have a stronger conclusion: $(\mathbb{I},|0\rangle|0\rangle)\models (\forall q_1 q_2)\beta^\prime$, which means that after any joint preparation operation on $q_1$ and $q_2$ together, denoted by a quantum term $\tau$, the state $\rho=\llbracket\tau\rrbracket_\mathbb{I}(|0\rangle|0\rangle)$ satisfies $\beta^\prime$. It should be noted that state $\rho$ prepared by the joint operation $\tau$ can be an entanglement between $q_1$ and $q_2$. 
\end{exam}

The second example shows how some important notions in fault-tolerant quantum computation \cite{Gott98} can be conveniently described in our logic $\mathcal{QL}$:

\begin{exam}Let $q_1,q_2$ be two quantum variables. Two states $|\phi\rangle,|\Psi\rangle\in\hs_{q_1}\otimes\hs_{q_2}$ are global unitary equivalent if there exists a unitary $U$ on $\hs_{q_1}\otimes\hs_{q_2}$ such that $|\Psi\rangle=U|\Phi\rangle$. They are local unitary equivalent (respectively, local Clifford equivalent), if there exist unitariies (respectively, Clifford operators) $U_1$ on $\hs_{q_1}$ and $U_2$ on $\hs_{q_2}$ such that $|\Psi\rangle=(U_1\otimes U_2)|\Phi\rangle$.   
Consider an interpretation $\mathbb{I}$ where quantum predicate symbol $P$ is interpreted as the $1$-dimensional subspace of $\hs_{q_1}\otimes\hs_{q_2}$ spanned by state $|\Psi\rangle$, and all quantum operation symbols are interpreted as unitary operators. Then global and local unitary equivalences of $|\Phi\rangle$ and $|\Psi\rangle$ can be expressed as \begin{align}\label{g-equivalence}&(\mathbb{I},|\Phi\rangle)\models(\exists q_1q_2)P(q_1,q_2),\\
\label{l-equivalence}&(\mathbb{I},|\Phi\rangle)\models(\exists q_1)(\exists q_2)P(q_1,q_2),\end{align} respectively. If all quantum operation symbols are interpreted as Clifford operators, then equation (\ref{l-equivalence}) expresses local Clifford equivalence of $|\Phi\rangle$ and $|\Psi\rangle$. It was conjectured in \cite{prob-list} but disproved in \cite{Ji10} that local unitary equivalence and local Clifford equivalence are equivalent
\end{exam}

As in classical first-order logic, we can show that the semantics of a logical formula in $\mathcal{QL}$ depends only on its free variables $\mathit{free}(\beta)$. 

\begin{defn} An interpretation $\mathbb{I}$ is called term-expressive if for any quantum variables $\overline{q}$ and for any $\rho,\rho^\prime$ and for any $\epsilon>0$, there exists a quantum term $\tau$ such that $\mathit{var}(\tau)\subseteq\overline{q}$ and $D\left(\llbracket\tau\rrbracket_\mathbb{I}(\rho),\rho^\prime\right)\leq\epsilon,$ where $D$ stands for the trace distance; that is, $D(\sigma,\sigma^\prime)=\frac{1}{2}\tr|\sigma-\sigma^\prime|$ for any density operators $\sigma,\sigma^\prime.$ 
\end{defn}

For example, the usual interpretation $\mathbb{I}$ of $H,S,T,\mathit{CNOT}$ as the Hadamard, phase, $\pi/8$ and controlled-NOT gates is term-expressive.  

\begin{lem}[Coincidence and Renaming]\label{lem-coincide}\begin{enumerate}
\item For any term-expressive interpretation $\mathbb{I}$, if $\rho\downarrow\mathit{free}(\beta)=\rho^\prime\downarrow\mathit{free}(\beta)$, then $(\mathbb{I},\rho)\models\beta$ iff $(\mathbb{I},\rho^\prime)\models\beta$. If $\beta$ is quantifier-free, then the requirement of term-expressivity is unnecessary.   
\item For any formula $\beta$, assume that $q\notin\mathit{free}(\beta)$, $q^\prime$ does not occur in $\beta$, and their dimensions are the same. If $\beta^\prime$ is obtained by replacing all bound occurrences of $q$ in $\beta$ with $q^\prime$, then $\beta\equiv \beta^\prime.$  
\end{enumerate}\end{lem}

The next lemma indicates that satisfaction relation is preserved by inclusion relation, convex combination and limit of quantum states. Essentially, this property comes from the linearity of quantum mechanics and quantum predicates as close subspaces.

\begin{lem}[Monotonicity, Convex combination and Limit]\label{lem-convex}\begin{enumerate}\item If $\supp\rho\subseteq\supp\sigma$ and $(\mathbb{I},\sigma)\models\beta$, then $(\mathbb{I},\rho)\models\beta$.   
\item If for each $i$, $(\mathbb{I},\rho_i)\models\beta$, then for any probability distribution $\{p_i\}$, we have $\left(\mathbb{I},\sum_i p_i\rho_i\right)\models\beta.$
\item If $(\mathbb{I},\rho_n)\models\beta$ for all $n$, and $\lim_{n\rightarrow}\rho_n=\rho$ with the trace distance, then $(\mathbb{I},\rho)\models\beta$.
\end{enumerate}\end{lem}

The notions of logical validity and consequence in classical first-order logic can be straightforwardly generalised into $\mathcal{QL}$. A formula $\beta$ is called valid in an interpretation $\mathbb{I}$, written $\mathbb{I}\models\beta$, if $(\mathbb{I},\rho)\models\beta$ for all $\rho$. The set of formulas valid in $\mathbb{I}$ is denoted $\mathit{Th}(\mathbb{I})=\{\beta|\mathit{I}\models\beta\}.$  A formula $\beta$ is called logical valid if it is valid in any interpretation $\mathbb{I}$. Let $\Sigma$ be a set of formulas. A formula $\beta$ is called a logical consequence of $\Sigma$, written $\Sigma\models\beta$, if for any interpretation $\mathbb{I}$ and for any state $\rho$, whenever $(\mathbb{I},\rho)\models\gamma$ for all $\gamma\in\Sigma$, then $(\mathbb{I},\rho)\models\beta$.     
Two formulas $\beta$ and $\beta^\prime$ are called 
logically equivalent, written $\beta\equiv\beta^\prime$, if $\beta\models\beta^\prime$ and $\beta^\prime\models\beta$. 

\subsubsection{Logical Semantics as Subspaces} As we saw in Subsections \ref{subsec-pql} and \ref{Sec-QL-FO}, in the propositional quantum logic and first-order quantum logic with classical variables, the semantics of a logical formula is defined as a closed subspace of a Hilbert space (or more generally, an element of an orthomodular lattice). However, the semantics of logical formulas in our logic $\mathcal{QL}$ is given in Definition \ref{QFO-semantics} in terms of satisfaction relation. 
On the other hand, for a classical first-order logical formula $\beta$ and an interpretation $\mathbb{I}$ with domain $D$, we have:\begin{equation}\label{two-semantics}\llbracket\beta\rrbracket_\mathbb{I}=\{d\in D:(\mathbb{I},d)\models\beta\}.\end{equation}
It is similar to the case of classical logic that we can establish a close connection between the subspace semantics and the satisfaction relation. To do so, let us first introduce the following definition as a quantum generalisation of equation (\ref{two-semantics}):

\begin{defn}\label{def-subspace-semantics}Given an interpretation $\mathbb{I}$. The semantics of a formula $\beta$ is defined the closed subspace $\llbracket\beta\rrbracket_\mathbb{I}$ of $\hs_\mathit{Var}$:
$$\llbracket\beta\rrbracket_\mathbb{I}=\bigvee_{(\mathbb{I},\rho)\models\beta}\supp \rho \in S(\hs_\mathit{Var})$$ where $\supp \rho$ denotes the support of density operator $\rho$, i.e. the subspace spanned by the eigenvectors of $\rho$ corresponding to its nonzero eigenvalues. 
\end{defn}

A close connection between the satisfaction in Definition \ref{QFO-semantics} and the subspace semantics in Definition \ref{def-subspace-semantics} is presented in the following:

\begin{lem}\label{subspace-satisfaction} For any formula $\beta$, interpretation $\mathbb{I}$ and quantum state $\rho$, $(\mathbb{I},\rho)\models\beta\ {\rm iff}\ \rho\in\llbracket\beta\rrbracket_\mathbb{I}.$  
\end{lem}

Furthermore, the subspace semantics of $\mathcal{QL}$ formulas enjoys a structural representation:

\begin{thm}\label{thm-subspace-satisfaction}
\begin{enumerate}\item If $\beta=P(\tau)$, then $\llbracket\beta\rrbracket_\mathbb{I}=\llbracket\tau\rrbracket_\mathbb{I}^\ast(P^\mathbb{I})\otimes\hs_{\mathit{Var}\setminus\mathit{var}(\tau)}$;
\item If $\beta=\neg\beta^\prime$, then $\llbracket\beta\rrbracket_\mathbb{I}=\left(\llbracket\beta^\prime\rrbracket_\mathbb{I}\right)^\perp$;
\item If $\beta=\beta_1\wedge\beta_2$, then $\llbracket\beta\rrbracket_\mathbb{I}=\llbracket\beta_1\rrbracket_\mathbb{I}\cap \llbracket\beta_2\rrbracket_\mathbb{I}$;
\item If $\beta=\tau^\ast(\beta^\prime)$, then $\llbracket\beta\rrbracket_\mathbb{I}=\llbracket\tau\rrbracket_\mathbb{I}^\ast(\llbracket\beta^\prime\rrbracket_\mathbb{I})$;
\item If $\beta=(\forall\overline{q})\beta^\prime$, then $\llbracket\beta\rrbracket_\mathbb{I}=\bigcap_{\mathit{var}(\tau)\subseteq\overline{q}}\left(\llbracket\tau\rrbracket^\ast_\mathbb{I}(\llbracket\beta^\prime\rrbracket\right).$
\end{enumerate}
\end{thm}

The clauses (2) and (3) in the above theorem indicates that the interpretation of propositional connectives $\neg$ and $\wedge$ in $\mathcal{QL}$ coincides with that in the original Birkhoff-von Neumann quantum logic. As a corollary of clauses (2) and (5) in the above theorem, if $\beta=(\exists\overline{q})\beta^\prime$, then
\begin{equation}\label{sem-exists}\llbracket\beta\rrbracket_\mathbb{I}=\bigvee_{\mathit{var}(\tau)\subseteq\overline{q}}\left(\llbracket\tau\rrbracket^\ast_\mathbb{I}(\llbracket\beta^\prime\rrbracket^\bot\right)^\bot.\end{equation} 

\subsection{Axiomatic System of $\mathcal{QL}$}

In this section, we present an axiomatisation of our logic $\mathcal{QL}$ with quantum variables. We promised in the Introduction that $\mathcal{QL}$ is a logic with equality $=$ for specifying and reasoning about equality of quantum states and equivalence of quantum circuits. But equality $=$ has not been introduced into $\mathcal{QL}$ in the previous subsections. Now we introduce $=$ in the following way: the axiomatic system of $\mathcal{QL}$ is designed as a two-layer system: the first layer is a first-order equational logic QT$_=$ for quantum terms, and the second layer is built upon QT$_=$ and consists of a set of inference rules for reasoning about first-order logical formulas with quantum variables.    

\subsubsection{Equational Logic for Quantum Terms}\label{subsec-equational}
Let us first describe logic QT$_=$. The formulas in QT$_=$ and their free variables are defined as follows: \begin{enumerate}
\item If $\tau_1$ and $\tau_2$ are quantum terms, then $\gamma=\tau_1=\tau_2$ is a formula in QT$_=$ and $\mathit{free}(\gamma)=\mathit{var}(\tau_1)\cup\mathit{var}(\tau_2)$;
\item If $\gamma^\prime$ are formulas in QT$_=$, so is $\gamma=\neg\gamma^\prime$, and $\mathit{free}(\gamma)=\mathit{free}(\gamma^\prime)$;
\item If $\gamma_1,\gamma_2$ are formulas in QT$_=$, so is $\gamma=\gamma_1\wedge\gamma_2,$ and $\mathit{free}(\gamma)=\mathit{free}(\gamma_1)\cup \mathit{free}(\gamma_2)$;
\item If $\gamma^\prime$ is a formula in QT$_=$, and $\overline{q}$ is a sequence of quantum variables, then $\gamma=(\forall\overline{q})\gamma^\prime$ is a formula in QT$_=$, and $\mathit{free}(\gamma)=\mathit{free}(\gamma^\prime)\setminus\overline{q}$.  
\end{enumerate} 
For any interpretation $\mathbb{I}$, quantum state $\rho$ and logical formula $\gamma$ in QT$_=$, satisfaction relation $(\mathbb{I},\rho)\models\gamma$ is defined as follows:
\begin{enumerate}\item If $\gamma=\tau_1=\tau_2$, then $(\mathbb{I},\rho)\models\gamma$ iff $\llbracket\tau_1\rrbracket_\mathbb{I}(\rho)=\llbracket\tau_2\rrbracket_\mathbb{I}(\rho)$;
\item If $\gamma=\neg\gamma^\prime$, then $(\mathbb{I},\rho)\models\gamma$ iff it does not hold that $(\mathbb{I},\rho)\models\gamma^\prime;$
\item If $\gamma=\gamma_1\wedge\gamma_2$, then $(\mathbb{I},\rho)\models\gamma$ iff $(\mathbb{I},\rho)\models\gamma_1$ and $(\mathbb{I},\rho)\models\gamma_2$;
\item If $\gamma=(\forall\overline{q})\gamma^\prime$, then $(\mathbb{I},\rho)\models\gamma$ iff for any quantum term $\tau$ with $\mathit{var}(\tau)\subseteq\overline{q}$, $(\mathbb{I},\llbracket\tau\rrbracket_\mathbb{I}(\rho))\models\gamma^\prime$.  
\end{enumerate} 
Similar to clause (5) in Definition \ref{QFO-semantics}, the semantics of universal quantification in QT$_=$ defined in clause (4) in the above definition follows the idea of equation (\ref{quantum-quantifier}). 

The axiomatic system of QT$_=$ consists of the standard inference rules of the classical first-order equational logic together with the rules given in Figure \ref{fig EQT}. We assume a special quantum operation symbol $I$ for the identity operator. 
 \begin{figure}[h]\centering
\begin{equation*}\begin{split}
&({\rm QT}1)\ \ \ \frac{\tau_1=\tau_2}{\tau\tau_1=\tau\tau_2}\qquad\frac{\tau_1=\tau_2}{\tau_1\tau=\tau_2\tau}\qquad\qquad ({\rm QT}2)\ \ \ \frac{\tau_{1i}=\tau_{i2}}{\sum_ip_i\tau_{1i}= \sum_ip_i\tau_{2i}}\\ 
&({\rm QT}3)\ \ \ \frac{\mathit{var}(\tau_1)\cap \mathit{var}(\tau_2)=\emptyset}{\tau_1\otimes\tau_2=\tau_1\tau_2=\tau_2\tau_1}\qquad\qquad\quad\  
({\rm QT}4)\ \ \ I\tau=\tau I=\tau\\ &({\rm QT}5)\ \ \ \tau_1(\tau_2\tau_3)=(\tau_1\tau_2)\tau_3\qquad\qquad\qquad\ \ ({\rm QT}6)\ \ \ \frac{\tau\ {\rm is\ unitary}}{\tau\tau^{-1}=\tau^{-1}\tau=I}
\end{split}\end{equation*}
\caption{Equational Logic QT$_=$ for Quantum Terms.}\label{fig EQT}
\end{figure}

It is easy to prove the following: 

\begin{lem}[Soundness]\label{Thm-QT-sound}If $\Gamma\vdash\tau_1=\tau_2$ is provable in the equational logic QT$_=$, then $\Gamma\models\tau_1=\tau_2.$ \end{lem}

 \subsubsection{Quantum Predicate Calculus}

The second layer of the axiomatic system of $\mathcal{QL}$ builds upon propositional quantum logic described in Subsection \ref{subsec-pql} and the equational logic QT$_=$ for quantum terms. It consists of propositional axiom ($\mathcal{QL}$1), equality axioms, convex combination axiom, term-adjoint axioms and quantifier axioms presented in Figure \ref{fig AIR}.
\begin{figure}[h]\centering
\begin{equation*}\begin{split}
&(\mathcal{QL}1)\ \ \ {\rm Any}\ \Sigma\vdash\beta\ {\rm provable\ in\ propositional\ QL\ given\ in\ Figure\ \ref{fig QL-P}}\\
&(\mathcal{QL}2)\ \ \ \frac{{\rm QT}_=\vdash\tau_1=\tau_2}{P(\tau_1)\vdash P(\tau_2)}\qquad\qquad\qquad\qquad\qquad\ \ \ \ \ \ \ \ (\mathcal{QL}3)\ \ \ \frac{{\rm QT}_=\vdash\tau_1=\tau_2}{\tau_1^\ast(\beta)\vdash \tau_2^\ast(\beta)}\\ 
&(\mathcal{QL}4)\ \ \ \frac{P(\tau_i)\ {\rm for\ all}\ i}{P\left(\sum_ip_i\tau_i\right)}\qquad\qquad\qquad\qquad\qquad\ \ \ \ \ \ \ \ \ \ (\mathcal{QL}5)\ \ \ \tau_1^\ast(\tau_2^\ast(\beta))\equiv(\tau_2\tau_1)^\ast(\beta)\\
&(\mathcal{QL}6)\ \ \ \frac{\beta\vdash\beta^\prime}{\tau^\ast(\beta)\vdash\tau^\ast(\beta^\prime)}\qquad\qquad\qquad\qquad\qquad\ \ \ \ \ \ \ \ \ (\mathcal{QL}7)\ \ \ \tau_1^\ast(P(\tau_2))\equiv P(\tau_1\tau_2)\qquad\qquad\\
&(\mathcal{QL}8)\ \ \ \frac{\tau\ {\rm is\ unitary}}{\tau^\ast(\neg\beta)\equiv\neg\tau^\ast(\beta)} \qquad\qquad\qquad\ \ \ \ \ \ \ \ \ \ \ \ \ \ \ \ \ \ (\mathcal{QL}9)\ \ \ \tau^\ast(\beta_1\wedge\beta_2)\equiv\tau^\ast(\beta_1)\wedge\tau^\ast(\beta_2)\\ 
&(\mathcal{QL}10)\ \ \ \frac{\mathit{free}(\beta_i)\subseteq\mathit{var}(\tau_i)\ {\rm for}\ i=1,2\qquad \mathit{var}(\tau_1)\cap\mathit{var}(\tau_2)=\emptyset}{(\tau_1\otimes\tau_2)^\ast(\beta_1\wedge\beta_2)\equiv\tau_1^\ast(\beta_1)\wedge\tau_2^\ast(\beta_2)}\\ 
&(\mathcal{QL}11)\ \ \ \frac{\tau^\ast(\beta)\vdash\gamma\quad \tau\ {\rm is\ unitary}}{\beta\vdash(\tau^{-1})^\ast\gamma}\qquad\qquad\qquad\quad\ \ \ \ (\mathcal{QL}12)\frac{\beta\vdash\tau^\ast(\gamma)\quad \tau\ {\rm is\ unitary}}{(\tau^{-1})^\ast(\beta)\vdash\gamma}\\
&(\mathcal{QL}13)\ \ \ \frac{\tau\ {\rm is\ unitary}\quad \mathit{var}(\tau)\subseteq\mathit{free}(\beta)\setminus \overline{q}}{\tau^\ast((\forall\overline{q})\beta)\equiv(\forall\overline{q})\tau^\ast(\beta)}\qquad\quad\ \ \
(\mathcal{QL}14)\ \ \ \frac{\mathit{var}(\tau)\subseteq\overline{q}}{\Sigma\cup\{(\forall\overline{q})\beta\}\vdash\tau^\ast(\beta)}\\
&(\mathcal{QL}15)\ \ \ \frac{\Sigma\vdash\beta\quad \overline{q}\cap\mathit{free}(\beta)=\emptyset\ {\rm or}\ \mathit{free}(\Sigma)\subseteq\mathit{free}(\beta)\setminus\overline{q}}{\Sigma\vdash(\forall\overline{q})\beta}
\end{split}\end{equation*}
\caption{Axiomatic System of $\mathcal{QL}$.}\label{fig AIR}
\end{figure}

The soundness of the axiomatic system is presented in the following theorem, but its completeness is still an open problem.   

\begin{thm}[Soundness]\label{QL-sound} If $\Sigma\vdash\beta$ is provable in the axiomatic system of $\mathcal{QL}$ given in Figure \ref{fig AIR}, then $\Sigma\models\beta$. 
\end{thm}

\section{Quantum Programs}\label{sec-QP} 

The first-order logic $\mathcal{QL}$ with quantum variables was established in the previous section. Now we turn to consider how it can be used as an assertion language for quantum programs. Let us first recast the syntax and semantics of a quantum programming language in the context of $\mathcal{QL}$. We choose to use the quantum extension of \textbf{while}-language defined in \cite{Ying16} because QHL (quantum Hoare logic) can be elegantly presented upon it.    

\subsection{Syntax}\label{sec-prog-syntax}
A classical \textbf{while}-language usually builds upon a first order logic, which is used to define, e.g. the term $t$ in an assignment $x:=t$ and the logical formula $b$ in a conditional statement $\mathbf{if}\ b\ \mathbf{then}\ S_1\ \mathbf{else}\ S_2$ or a \textbf{while}-statement $\mathbf{while}\ b\ \mathbf{do}\ S\ \mathbf{od}$. In the original definition of quantum \textbf{while}-language given in \cite{Ying16}, however, the corresponding parts were left undefined formally due to the lack of an appropriate first-order logic. Now we are able to fill in this gap by incorporating logic $\mathcal{QL}$ into the syntax of quantum \textbf{while}-language. 
In this vein, the alphabet of quantum \textbf{while}-language consists of: \begin{itemize}\item[(i)] a set $\mathit{Var}$ of quantum variables $q,q_1,q_2...$;
 \item[(ii)] a set of unitary symbols $U,U_1,U_2,...$ (for example, a universal set of basic gates:  Hadamard gate $H$, phase gate $S$, $\pi/8$ gate $T$ and ${\rm CNOT}$ - controlled NOT);
 \item[(iii)] a set of measurement symbols $M,M_1,M_2,...$; and 
\item[(vi)] program constructors $:=$ (for initialisation and unitary transformations), $;$ (sequential composition), $\mathbf{if}...\mathbf{fi}$ (case statement); $\mathbf{while}...\mathbf{do}...\mathbf{od}$ (loop).\end{itemize} 

As in the alphabet of logic $\mathcal{QL}$, each variable $q\in\mathit{Var}$ is associated with a nonnegative integer $d$ or $d=\infty$ as its dimension, and each unitary symbol $U$ is associated with an $n$-tuple $(d_1,...,d_n)$ as its signature, where $d_1,...,d_n$ are nonnegative integers or $\infty$, and $n$ is a nonnegative integer, called the arity of $U$. An arity $n$ and a signature $(d_1,...,d_n)$ are also assigned to each measurement symbol $M$. In addition, a set $\mathit{out}(M)$ is assigned to measurement symbol $M$ and stands for the set of all possible outcomes of the measurement denoted by $M$. 

We now can define quantum terms over the above alphabet. To this end, let us introduce some additional quantum operation symbols:
\begin{itemize}\item For any nonnegative integer $d$ or $d=\infty$, we use $\mathbf{0}$ to denote the initialisation of a $d$-dimensional quantum variable in the basis state $|0\rangle$;
\item For each measurement symbol $M$ and for each $m\in\mathit{out}(M)$, we use $M_m$ for the operation that the measurement denoted by $M$ is performed, the outcome $m$ is observed and the state of the measured system is changed accordingly.  
\end{itemize}
Then basic terms includes:\begin{itemize}\item 
 $\mathbf{0}_d(q)$ for each $d$-dimensional quantum variable $q\in\mathit{Var}$, meaning that $q$ is initialised in state $|0\rangle$. It is often simply written as $\mathbf{0}(q)$; 
\item $U(\overline{q}),$ where variables $\overline{q}=q_1...q_n$ match the signature of $U$; that is, if $U$ has signature $(d_1,...,d_n)$, then $q_i$ is $d_i$-dimensional for each $i$;
\item $M_m(\overline{q})$, where variables $\overline{q}=q_1,...,q_n$ match the signature of $M$.
\end{itemize}
Upon them, all quantum terms can be constructed by applying the formation rules (2) - (4) in Definition \ref{def-terms}.
Furthermore, using the notion of quantum term, the syntax of quantum \textbf{while}-language can be restated in the following:   

\begin{defn}Quantum programs are defined by
the syntax:
\begin{align*}S::=\ \mathbf{skip}\ & \ |\ q:=\mathbf{0}(q)\ |\ \overline{q}:=\tau\ |\ S_1;S_2\ |\ \mathbf{if}\ \left(\square m\cdot M[\overline{q}] =m\rightarrow S_m\right)\ \mathbf{fi}\ |\ \mathbf{while}\ M[\overline{q}]=1\ \mathbf{do}\ S\ \mathbf{od}\end{align*}
where $\tau$ is a unitary term with $\mathit{var}(\tau)\subseteq\overline{q}$, and $M$ in the case statement (respectively, the loop) is a measurement symbol with $\mathrm{out}(M)=\{m\}$ (respectively, $\mathrm{out}(M)=\{0,1\}$). \end{defn}

For each program $S$, we write $\mathit{var}(S)$ for the set of quantum variables in $S$.  

\begin{rem} The state-of-the-art quantum hardware is the so-called NISQ (Noisy Intermediate Scale Quantum) devices. So, several papers (e.g. \cite{Wu19, Gu21}) have been devoted to analysis of noisy quantum programs where noise may occur in quantum gates. To include these programs, the above definition can be extended by allowing $\tau$ to be quantum terms that are not unitary.\end{rem}

\subsection{Semantics}\label{sec-prog-semantics}

Based on the semantics of first order logic $\mathcal{QL}$ defined in the last section, the semantics of quantum \textbf{while}-programs can also be formulated in a more precisely way than that originally given in \cite{Ying16}. An interpretation $\mathbb{I}$ of quantum \textbf{while}-language is given as follows:\begin{itemize}\item To each $d$-dimensional quantum variable $q$, a $d$-dimensional Hilbert space $\hs_q$ is assigned, called the state space of $q$;
\item Each unitary symbol $U$ with signature $(d_1,...,d_n)$ is interpreted as a $\prod_{i=1}^nd_i$-dimensional unitary operator $U^\mathbb{I}$; and 
 \item Each measurement symbol $M$ with signature $(d_1,...,d_n)$ and outcomes $\mathit{out}(M)$ is interpreted as a \textit{projective} measurement $M^\mathbb{I}=\{M^\mathbb{I}_m:m\in\mathit{out}(M)\}$ on the $\prod_{i=1}^nd_i$-dimensional Hilbert space, where each $M^\mathbb{I}_m$ is a projection operator.  
 \end{itemize} 

\begin{rem}  In this paper, since we only consider sharp quantum logic with projection operators (equivalently, closed subspaces) as quantum proposition, a measurement symbol $M$ is always interpreted as a projective measurement. But this does not reduce the expressive power of our programming language because a general measurement can be expressed in terms of a projective measurement together with a unitary transformation and some ancillary variables.  
\end{rem}

The semantics of quantum terms in logic $\mathcal{QL}$ can be directly applied here. Given an interpretation $\mathbb{I}$, we write $\hs_V=\bigotimes_{q\in V}\hs_q$  for the state space of the composed system of quantum variables in $V$. Then each term $\tau$ is interpreted as a mapping $\llbracket\tau\rrbracket_\mathbb{I}:\mathcal{D}(\hs_\mathit{Var})\rightarrow\mathcal{D}(\hs_\mathit{Var})$. In particular, the basic terms are interpreted as follows: for any $\rho\in\mathcal{D}(\hs_\mathit{Var})$, 
\begin{itemize}
\item $\llbracket 0(q)\rrbracket_\mathbb{I}(\rho)=\sum_i|0\rangle_q\langle i|\rho|i\rangle_q\langle 0|,$ where $\{|i\rangle\}$ is an orthonormal basis of $\hs_q$;
\item $\llbracket U(\overline{q})\rrbracket_\mathbb{I}(\rho)=(U^\mathbb{I}\otimes I)\rho((U^\mathbb{I})^\dag\otimes I),$ where $I$ is the identity operator on $\hs_{\mathit{Var}\setminus\overline{q}}$;
\item $\llbracket M_m(\overline{q})\rrbracket_\mathbb{I}(\rho)=(M_m^\mathbb{I}\otimes I)\rho((M_m^\mathbb{I})^\dag\otimes I),$ where $I$ is the same as above. 
\end{itemize} For other terms $\tau$, its semantics $\llbracket\tau\rrbracket_\mathbb{I}$ is defined using valuation rules (2) - (4) in Definition \ref{def-sem-terms}. 

Now we can define the semantics of quantum programs based on the semantics of quantum terms. 
 A configuration is defined as a pair $C=\langle S,\rho\rangle$, 
where $S$ is a program or the termination symbol $\downarrow$, and $\rho\in\mathcal{D}(\mathcal{H}_\mathit{Var})$ denotes a state of quantum variables. 

\begin{defn}The operational semantics of quantum programs is a transition relation between configurations defined by the transition rules in Figure \ref{fig QP-OP}. \begin{figure}[h]\centering
\begin{equation*}\begin{split}&({\rm Sk})\ \ \langle\mathbf{skip},\rho\rangle\rightarrow\langle \downarrow,\rho\rangle\qquad\qquad\qquad\qquad\ ({\rm In})\ \ \ \langle
q:=\mathbf{0}(q),\rho\rangle\rightarrow\langle \downarrow,\llbracket 0(q)\rrbracket_\mathbb{I}(\rho)\rangle\\
&({\rm UT})\ \ \langle\overline{q}:=\tau,\rho\rangle\rightarrow\langle
\downarrow, \llbracket\tau\rrbracket_\mathbb{I}(\rho)\rangle\qquad\qquad\ \ \ ({\rm SC})\ \ \ \frac{\langle S_1,\rho\rangle\rightarrow\langle
S_1^{\prime},\rho^{\prime}\rangle} {\langle
S_1;S_2,\rho\rangle\rightarrow\langle
S_1^{\prime};S_2,\rho^\prime\rangle}\\
&({\rm IF})\ \ \ \langle\mathbf{if}\ (\square m\cdot
M[\overline{q}]=m\rightarrow S_m)\ \mathbf{fi},\rho\rangle\rightarrow\langle
S_{m^\prime},\llbracket M_{m^\prime}(\overline{q})\rrbracket_\mathbb{I}(\rho)\rangle\ {\rm for\ every}\ m^\prime\in\mathit{out}(M)\\
&({\rm L}0)\ \ \ \langle\mathbf{while}\
M[\overline{q}]=1\ \mathbf{do}\
S\ \mathbf{od},\rho\rangle\rightarrow\langle \downarrow, \llbracket M_0(\overline{q})\rrbracket_\mathbb{I}(\rho)\rangle\\
&({\rm L}1)\ \ \ \langle\mathbf{while}\
M[\overline{q}]=1\ \mathbf{do}\ S\ \mathbf{od},\rho\rangle\rightarrow
\langle S;\mathbf{while}\ M[\overline{q}]=1\ \mathbf{do}\ S\ \mathbf{od}, \llbracket M_1(\overline{q})\rrbracket_\mathbb{I}(\rho)\rangle\end{split}\end{equation*}
\caption{Operational Semantics of Quantum Programs.}\label{fig QP-OP}
\end{figure}\end{defn}

Transition rules (Sk) and (SC) are the same as in classical programming, and others are defined directly by the basic postulates of quantum mechanics. In particular, the transitions in (IF), (L0) and (L1) are essentially probabilistic; for example, for each $m^\prime$, the transition in (IF) happens with probability $p_{m^\prime}=\mathit{tr}(\llbracket M_{m^\prime}(\overline{q})\rrbracket_\mathbb{I}(\rho))=\mathit{tr}(M_{m^\prime}^\mathbb{I}\rho),$ and the program state is changed from $\rho$ to $\rho_{m^\prime}=\llbracket M_{m^\prime}(\overline{q})\rrbracket_\mathbb{I}(\rho) /p_{m^\prime}.$
For simplicity, following a convention suggested in \cite{Selinger04}, probability $p_{m^\prime}$ and density operator $\rho_{m^\prime}$ are combined into a partial density operator $\llbracket M_{m^\prime}(\overline{q})\rrbracket_\mathbb{I}(\rho)=p_{m^\prime}\rho_{m^\prime}$. 

\begin{defn}The denotational semantics of a quantum $S$ in interpretation $\mathbb{I}$ is a mapping $\llbracket S\rrbracket_\mathbb{I}:\mathcal{D}(\mathcal{H}_\mathit{Var})\rightarrow \mathcal{D}(\mathcal{H}_\mathit{Var})$ from density operators to (partial) density operators. It is defined by \begin{equation*}\llbracket S\rrbracket_\mathbb{I}(\rho)=\sum\left\{|\rho^\prime: \langle S,\rho\rangle\rightarrow^\ast\langle \downarrow,\rho^\prime\rangle|\right\}\end{equation*} for every $\rho\in\mathcal{D}(\mathcal{H}_\mathit{Var})$, where $\rightarrow^\ast$ is the reflexive and transitive closure of transition relation $\rightarrow$ (operational semantics), and $\left\{|\cdot|\right\}$ denotes a multi-set.\end{defn}

\section{Quantum Hoare Logic with $\mathcal{QL}$ as Its Assertion Language}\label{Sec-QHL} 

After redefining the syntax and semantics of quantum program upon first order logic $\mathcal{QL}$ with quantum variables, we are able to further incorporate $\mathcal{QL}$ into quantum Hoare logic (QHL) so that $\mathcal{QL}$ can serve as an assertion logic of QHL. With the help of $\mathcal{QL}$, QHL can be described in a more elegant way; in particular, its relative completeness of QHL can be precisely formulated.

\subsection{Correctness Formulas}
As said in the Introduction, in this paper, we only consider a simplified version of quantum Hoare logic \cite{Ying11} where preconditions and postconditions are modelled as projection operators (equivalently, closed subspaces), as described in \cite{Zhou19}. Using $\mathcal{QL}$, these preconditions and postconditions can be expressed by first-order logical formulas, and thus the notion of quantum Hoare triple and the correctness of quantum programs can be precisely defined as follows. 

\begin{defn}[Hoare Triple] A Hoare triple (or correctness formula) is a Hoare triple, i.e. a statement of the form: $\{\beta\}S\{\gamma\},$ where $S$ is a quantum \textbf{while}-program over the alphabet given in Subsection \ref{sec-prog-syntax}, and both $\beta, \gamma$
are logical formulas in $\mathcal{QL}$ over the same alphabet, called the precondition and postcondition, respectively.\end{defn}

It should be pointed out that the preconditions and postconditions in both the original quantum Hoare logic \cite{Ying16} and the simplified version presented in \cite{Zhou19} were not formally defined in a logical language. Instead, they were simply assumed to be Hermitian operators or projection operators, which are mathematical objects rather than logical formulas. The reason is again that an appropriate logic for specifying them was lacking at that time. The first order logic $\mathcal{QL}$ with quantum variables introduced in Section \ref{sec-QL-new} provides us with the logical tools required in the above definition.      

Let us now define the semantics of Hoare triples. For simplicity, we only consider partial correctness in this paper. Total correctness can be treated by adding certain termination condition as in classical programming. 
\begin{defn} Given an interpretation $\mathbb{I}$ of the language defined in Subsection \ref{sec-prog-syntax}. A Hoare triple $\{\beta\}S\{\gamma\}$ is true in
the sense of partial correctness in $\mathbb{I}$, written
$\models_{\rm par}^\mathbb{I}\{\beta\}S\{\gamma\},$ if we have: 
\begin{equation}\label{par-correct}\llbracket S\rrbracket_\mathbb{I}(\llbracket\beta\rrbracket_\mathbb{I})\subseteq \llbracket\gamma\rrbracket_\mathbb{I}
\end{equation} where $\llbracket S\rrbracket_\mathbb{I}$, $\llbracket\beta\rrbracket_\mathbb{I}$ and $\llbracket\gamma\rrbracket_\mathbb{I}$ are defined as in Subsection \ref{sec-prog-semantics}, and $\llbracket S\rrbracket_\mathbb{I}(\llbracket\beta\rrbracket_\mathbb{I})$ is the image of subspace $\llbracket\beta\rrbracket_\mathbb{I}$ under super-operator $\llbracket S\rrbracket_\mathbb{I}$ as defined by equation (\ref{definition-image}). 
 \end{defn}
 
 For any $X\in\mathcal{S}(\hs_\mathit{all})$ and $\rho\in\mathcal{D}(\hs_\mathit{all})$, we say that $\rho$ belongs to $P$, written
$\rho\in P$, if $\supp(\rho)\subseteq X$.
Then condition (\ref{par-correct}) can be restated in a more intuitive way: for all 
$\rho,$ $\rho\in \llbracket\beta\rrbracket\ {\rm implies}\ \llbracket S\rrbracket_\mathbb{I} (\rho)\in \llbracket\gamma\rrbracket.$
 
 \begin{defn}Let $\Sigma$ be a set of logical formulas in the assertion language. Then a Hoare triple $\{\beta\}S\{\gamma\}$ is called a logical consequence of $\Sigma$ in the sense of partial correctness, written $\Sigma\models_\mathit{par}\{\beta\}S\{\gamma\},$ 
if for any interpretation $\mathbb{I}$, we have:
$${\rm whenever\ all\ formulas\ in}\ \Sigma\ {\rm are\ true\ in}\ \mathbb{I},\ {\rm then}\ \models_\mathit{par}^\mathbb{I}\{\beta\}S\{\gamma\}.$$
 \end{defn}

\subsection{Axiomatic System}\label{sec-QHL-axioms}
An axiomatic system of quantum Hoare logic was first presented in \cite{Ying11} for the general case where preconditions and postconditions can be any quantum predicates, i.e. Hermitian operators between the zero and identity operators. It was slightly simplified in \cite{Zhou19} for the special case where preconditions and postconditions are restricted to projection operators. 
Using $\mathcal{QL}$ as the assertion language, the axiomatic system of \cite{Zhou19} can be recasted in Figure \ref{fig axioms} in a more elegant way. Note that quantum logical connectives $\wedge,\vee$ are used in the rules (R.IF) and (R.LP), and term-adjoint formulas are used in the axiom (Ax.UT). Moreover, entailment in $\mathcal{QL}$ is employed in the rule (R.Con).  

\begin{figure}[h]\centering
\begin{equation*}\begin{split}
&({\rm Ax.Sk})\ \ \ \{\beta\}\mathbf{Skip}\{\beta\}\qquad\qquad \qquad\qquad\ \ \ \ ({\rm Ax.In})\ \ \ \left\{\mathbf{0}(q)^\ast\beta \right\}q:=|0\>\{\beta\}\\
&({\rm Ax.UT})\ \ \
\left\{\tau^\ast(\beta)\right\}\overline{q}:=\tau\{\beta\}\qquad\qquad\qquad ({\rm R.SC})\ \ \
\frac{\{\beta\}S_1\{\gamma\}\ \ \ \ \ \ \{\gamma\}S_2\{\delta\}}{\{\gamma\}S_1;S_2\{\delta\}}\\
&({\rm R.IF})\ \ \  
\frac{\left\{\beta_m\right\}S_m\{\gamma\}\ {\rm for\ all}\ m}{\left\{\bigvee_m(M_m(\overline{q})\wedge \beta_m)\right\}\mathbf{if}\ (\square m\cdot M[\overline{q}] = m \rightarrow S_m )\ \mathbf{fi}\{\gamma\}}\\
&({\rm R.LP})\ \ \
\frac{\{\beta\}S\{(M_0(\overline{q})\wedge \gamma)\vee (M_1(\overline{q})\wedge \beta)\}}{\{(M_0(\overline{q})\wedge \gamma)\vee (M_1(\overline{q})\wedge \beta)\}\mathbf{while}\
M[\overline{q}]=1\ \mathbf{do}\ S\ \mathbf{od}\{\gamma\}}\\
&({\rm R.Con})\ \ \ \frac{\beta\vdash
\beta^{\prime}\ {\rm in}\ \mathcal{QL}\ \ \ \ \{\beta^{\prime}\}S\{\gamma^{\prime}\}\ \ \ \
\gamma^{\prime}\vdash \gamma\ {\rm in}\ \mathcal{QL}}{\{\beta\}S\{\gamma\}}
\end{split}\end{equation*}
\caption{Axiomatic System of Quantum Hoare Logic (for Partial Correctness).}\label{fig axioms}
\end{figure}

It should be particularly pointed out that the axioms and inference rules of quantum Hoare logic are valid in all interpretations and not designed for any specific interpretation. Therefore, to prove the correctness of a quantum program in a specific interpretation $\mathbb{I}$ using quantum Hoare logic, one may need to call upon some logical formulas from the theory $\mathrm{Th}(\mathbb{I})$ of\ $\mathbb{I}$, i.e. the set of all logical formulas $\beta$ in $\mathcal{QL}$ that are true in $\mathbb{I}$. For example, let $q$ be a qubit variable, $H$ the Hadamard gate, and $X$ an arbitrary subspace of the $2$-dimensional Hilbert space. Then the correctness $\models_\mathit{par}\{X\}q:=H[q];q:=H[q]\{X\}$ can be verified using the axiom (Ax.UT) together with the specific property $HH=I$ of the Hadamard gate, but not by quantum Hoare logic solely.   

\subsection{Soundness and Relative Completeness}

The soundness and relative completeness of quantum Hoare logic (QHL) with general quantum predicates was established in \cite{Ying11}. Then the soundness and relative completeness of QHL for the special quantum predicates of projection operators    
, which is exactly the case considered in this paper, was derived in \cite{Zhou19} through a simple reduction from that of original QHL in \cite{Ying11}. 
However, due the lacking of a precisely defined assertion language, the soundness and relative completeness was only described in an informal way there. 
Now with the help of assertion $\mathcal{QL}$, we can present them in a formal way. As usual, the soundness is easy to prove.  

\begin{thm}[Soundness of QHL] For any set $\Sigma$ of logical formulas in the assertion language $\mathbf{QL}$, and for any Hoare triple $\{\beta\}S\{\gamma\}$:
$${\rm if}\ \Sigma\vdash_\mathit{par}\{\beta\}S\{\gamma\}\ {\rm then}\ \Sigma\models_\mathit{par}\{\beta\}S\{\gamma\}.$$
\end{thm}

As in the case of classical programming, it is easy to see that the inverse of the above soundness theorem is not true. To give a formal presentation of the relative completeness, let us introduce the following:

\begin{defn} An interpretation $\mathbb{I}$ is said to be expressive if for any program $S$ and for any logical formula $\gamma$ in $\mathcal{QL}$, there exists a logical formula $\beta$ in $\mathcal{QL}$ such that \begin{align*}
\llbracket\beta\rrbracket_\mathbb{I}&=\textrm{wlp}.\llbracket S\rrbracket_\mathbb{I}.\llbracket\gamma\rrbracket_\mathbb{I}\ {\rm (weakest\ liberal\ precondition)}\\
&=\overline{\mathrm{span}\left\{|\psi\rangle\in\hs_\mathit{all}:\llbracket S\rrbracket_\mathbb{I}(|\psi\rangle\langle\psi|)\in \llbracket\gamma\rrbracket_\mathbb{I}\right\}}.
\end{align*}
\end{defn}

With the help of the above definition, the relative completeness can be stated as the following: 

\begin{thm}[Relative Completeness of QHL] Let $\mathbb{I}$ be an expressive interpretation. Then for any Hoare triple $\{\beta\}S\{\gamma\}$:
$${\rm if}\ \models^{\mathbb{I}}_\mathit{par}\{\beta\}S\{\gamma\}\ {\rm then}\ \mathrm{Th}(\mathbb{I})\vdash_\mathit{par}\{\beta\}S\{\gamma\}$$ where $\mathrm{Th}(\mathbb{I})$ is the theory of\ $\mathbb{I}$, i.e. the set of all logical formulas $\beta$ in $\mathcal{QL}$ that are true in $\mathbb{I}$.  
\end{thm}

\begin{proof} (Outline) It suffices to reformulate the proof of relative completeness given in \cite{Ying11,Zhou19} in the logical language of $\mathcal{QL}$. \end{proof}

An interesting problem that remains open is to determine a set of quantum operations (including unitary operators and measurement) commonly used in quantum computing (e.g. Hadamard gate, CNOT, measurement in the computational basis) that is expressible. This problem seems not easy if we consider not only finite-dimensional quantum variables but also infinite-dimensional ones with the Hilbert space defined in (\ref{infinite-Hilbert}). 

\section{Applications}\label{sec-Applications}

As applications of the theoretical results developed in the previous sections, in this section, we give some examples to demonstrate how can first order logic $\mathcal{QL}$ with quantum variables as an assertion logic and QHL as a program logic can work together to reason about quantum programs in a more convenient and economic way.   

\subsection{Adaptation Rules}\label{sec-adaptation-rules}

As is well-known \cite{Apt09,Apt19} and discussed in Section \ref{Intro}, adaptation rules can often significantly simplify verification of classical programs and are expected to play the same role for quantum programs. The consequence rule (R.Con) in QHL given in the last subsection is an adaptation rule. In this subsection, as the first application of $\mathcal{QL}$, we show how it can be used to define more adaptation rules for quantum programs. To this end, let us first introduce:

\begin{defn}Let $\mathbb{I}$ be an interpretation. Then:\begin{enumerate}\item We say that a quantum program $S$ terminates in $\mathbb{I}$, written $\mathbb{I}\models S:{\rm Term}$, if $\llbracket S\rrbracket_\mathbb{I}(I)=I$, where $I$ is the identity operator. 
\item A quantum program $S$ is called term representable in $\mathbb{I}$ if there exists a quantum term $\tau$ such that $\mathit{var}(\tau)\subseteq\mathit{var}(S)$ and $\llbracket S\rrbracket_\mathbb{I}(\llbracket\tau \rrbracket_\mathbb{I}^\ast(X))=X$ for any $X\in S(\hs_{\mathit{var}(S)}).$   
\end{enumerate}
\end{defn}

A large number of adaptation rules have been introduced for classical programs in the literature. Here, we only generalise some of the most popular presented in Section 3.8 of \cite{Apt09} and Subsection 5.1 of \cite{Apt19} to the quantum case as examples showing the applicability of assertion logic $\mathcal{QL}$. The quantum generalisations of these rules are presented in Figure \ref{fig auxi-rules}. It should be noticed that the connectives of conjunction and disjunction in Birkhoff-von Neumann quantum logic are employed in the rules (Invariance), (Conjunction) and (Disjunction), and the term-adjoint formulas and quantifiers over quantum variables newly introduced in this paper are used in the rules ($\exists$-Introduction) and (Hoare Adaptation). In particular, the Hoare adaptation rule is crucial for reasoning about procedure calls and recursion \cite{Hoare71}, and has been extended for reasoning about method calls in object-oriented programs \cite{Apt19}. We expect that its quantum generalisation given in Figure \ref{fig auxi-rules} will play a similar role in quantum programming.    

\begin{figure}[h]\centering
\begin{equation*}\begin{split}
&({\rm Invariance})\ \ \ \frac{\{\beta\}S\{\gamma\}\quad \mathit{free}(\delta)\cap\mathit{var}(S)=\emptyset}{\{\beta\wedge\delta\}S\{\gamma\wedge\delta\}}\qquad ({\rm Substitution})\ \ \ \frac{\{\beta\}S\{\gamma\}\quad \mathit{var}(\tau)\cap\mathit{var}(S)=\emptyset}{\left\{\tau^\ast(\beta)\right\}S\left\{\tau^\ast(\gamma)\right\}}\\
&({\rm Conjunction})\ \ \  
\frac{\left\{\beta_1\right\}S\left\{\gamma_1\right\}\quad \left\{\beta_2\right\}S\left\{\gamma_2\right\}}{\left\{\beta_1\wedge\beta_2\right\}S\left\{\gamma_1\wedge\gamma_2\right\}}\qquad\qquad  
({\rm Disjunction})\ \ \
\frac{\left\{\beta_1\right\}S\{\gamma\}\quad \left\{\beta_2\right\}S\{\gamma\}}{\left\{\beta_1\vee\beta_2\right\}S\{\gamma\}}\\
&(\exists{\rm -Introduction})\ \ \ \frac{\{\beta\}S\{\gamma\}\quad \overline{q}\cap[\mathit{var}(S)\cap\mathit{free}(\gamma)]=\emptyset\quad S\ {\rm terminates}}{\{(\exists \overline{q})\beta\}S\{\gamma\}}\\
&({\rm Hoare\ Adaptation})\ \ \ \frac{\begin{split}\{\beta\}S\{\gamma\}\quad&\mathit{var}(S)\subseteq\overline{p}\quad \overline{q}=\mathit{free}(\beta)\cup\mathit{free}(\gamma)\setminus[\mathit{free}(\delta)\cup\overline{p}]\\ &S\ {\rm is\ term\ representable}\end{split}}{\{(\exists\overline{q})[\beta\wedge(\forall\overline{p})(\gamma\rightarrow\delta)]\}S\{\delta\}}
\end{split}\end{equation*}
\caption{Auxiliary Rules}\label{fig auxi-rules}
\end{figure}

\begin{thm}[Soundness of Adaptation Rules]\label{thm-auxi-rules} All of the proof rules in Figure \ref{fig auxi-rules} are sound in the sense of partial correctness.
\end{thm}

\begin{proof} Here, we choose to prove the soundness of (Hoare Adaptation) because several key laws in Birkhoff-von Neumann quantum logic are used in the proof in an essential way, including the ortho-modularity (see Theorem \ref{thm-Sasaki}). The soundness of other rules are proved in the Appendix. Assume that $\models_\mathit{par}\{\beta\}S\{\gamma\}$, $\mathit{var}(S)\subseteq\overline{p}$ and $\overline{q}=\mathit{free}(\beta)\cup\mathit{free}(\gamma)\setminus[\mathit{free}(\delta)\cup\overline{p}].$
We want to prove \begin{equation}\label{adapt-1f}\models_\mathit{par}\{(\exists\overline{q})[\beta\wedge(\forall\overline{p})(\gamma\rightarrow\delta)]\}S\{\delta\}.\end{equation} 

First, we have: for any interpretation $\mathbb{I}$,
\begin{align*}\llbracket S\rrbracket_\mathbb{I} (\llbracket (\exists\overline{q})[\beta\wedge(\forall\overline{p})(\gamma\rightarrow\delta)]\rrbracket_\mathbb{I} &= \llbracket S\rrbracket_\mathbb{I}\left(\bigvee_{\mathit{var}(\tau)\subseteq\overline{q}}\llbracket\tau \rrbracket_\mathbb{I}^\ast \left(\llbracket\beta \rrbracket_\mathbb{I}\cap \llbracket(\forall\overline{p})(\gamma\rightarrow\delta)\rrbracket_\mathbb{I}\right)\right)\\
&=\bigvee_{\mathit{var}(\tau)\subseteq\overline{q}}\llbracket S\rrbracket_\mathbb{I}\left(\llbracket\tau \rrbracket_\mathbb{I}^\ast(\llbracket\beta \rrbracket_\mathbb{I}\cap \llbracket(\forall\overline{p})(\gamma\rightarrow\delta) \rrbracket_\mathbb{I})\right)
\end{align*}
where the last equality is derived using Lemma A.5 (2). Therefore, it suffices to show that for any quantum term $\tau$ with $\mathit{var}(\tau)\subseteq\overline{q}$,
$\llbracket S\rrbracket_\mathbb{I}\left(\llbracket\tau \rrbracket_\mathbb{I}^\ast(\llbracket\beta \rrbracket_\mathbb{I}\cap \llbracket(\forall\overline{p})(\gamma\rightarrow\delta) \rrbracket_\mathbb{I})\right)\subseteq 
\llbracket\delta\rrbracket_\mathbb{I}.$ We observe: \begin{align}\label{adapt-3}\llbracket S\rrbracket_\mathbb{I}\left(\llbracket\tau \rrbracket_\mathbb{I}^\ast(\llbracket\beta \rrbracket_\mathbb{I}\cap \llbracket(\forall\overline{p})(\gamma\rightarrow\delta) \rrbracket_\mathbb{I})\right)&=\llbracket\tau \rrbracket_\mathbb{I}^\ast\left(\llbracket S\rrbracket_\mathbb{I}(\llbracket\beta \rrbracket_\mathbb{I}\cap \llbracket(\forall\overline{p})(\gamma\rightarrow\delta) \rrbracket_\mathbb{I})\right)\\ \label{adapt-4}&\subseteq \llbracket\tau \rrbracket_\mathbb{I}^\ast\left(\llbracket S\rrbracket_\mathbb{I}(\llbracket\beta \rrbracket_\mathbb{I}) \cap \llbracket S\rrbracket_\mathbb{I}(\llbracket(\forall\overline{p})(\gamma\rightarrow\delta) \rrbracket_\mathbb{I})\right)\\ \label{adapt-5}&\subseteq \llbracket\tau \rrbracket_\mathbb{I}^\ast\left(\llbracket\gamma \rrbracket_\mathbb{I} \cap \llbracket S\rrbracket_\mathbb{I}(\llbracket(\forall\overline{p})(\gamma\rightarrow\delta) \rrbracket_\mathbb{I})\right)\end{align}
Here, (\ref{adapt-3}) is true because $\mathit{var}(S)\cap\mathit{var}(\tau)=\emptyset$, which is implied by the assumption that $\mathit{var}(S)\subseteq\overline{p}$ and $\mathit{var}(\tau)\subseteq\overline{q}=\mathit{free}(\beta)\cup\mathit{free}(\gamma)\setminus[\mathit{free}(\delta)\cup\overline{p}];$ (\ref{adapt-4}) is obtained using Lemma A.5 (2); and (\ref{adapt-5}) is derived from the assumption $\models_\mathit{par}\{\beta\}S\{\gamma\}$, i.e.   
$\llbracket S\rrbracket_\mathbb{I}(\llbracket\beta \rrbracket_\mathbb{I})\subseteq\llbracket\gamma \rrbracket_\mathbb{I}.$

Let us assume that $S$ is term representable in $\mathbb{I}$. Then there exists a quantum term $\sigma_0$ such that $\mathit{var}(\sigma_0)\subseteq\mathit{var}(S)$ and 
$\llbracket S\rrbracket_\mathbb{I}(\llbracket\sigma_0 \rrbracket_\mathbb{I}(X))=X$ for all $X$. Furthermore, we have:
\begin{align*}\llbracket S\rrbracket_\mathbb{I}(\llbracket (\forall\overline{p})(\gamma\rightarrow\delta)\rrbracket_\mathbb{I}&=\llbracket S\rrbracket_\mathbb{I}\left(\bigcap_{\mathit{var}(\sigma)\subseteq\overline{p}}\llbracket \sigma\rrbracket_\mathbb{I}^\ast(\llbracket \gamma\rightarrow\rrbracket_\mathbb{I})\right)\subseteq \llbracket S\rrbracket_\mathbb{I}\left(\llbracket\sigma_0 \rrbracket_\mathbb{I}^\ast(\llbracket\gamma\rightarrow\delta \rrbracket_\mathbb{I})\right)
\\ &= \llbracket S\rrbracket_\mathbb{I}\left(\llbracket\sigma_0 \rrbracket_\mathbb{I}^\ast(\llbracket\gamma\rrbracket_\mathbb{I}\rightarrow\llbracket\delta \rrbracket_\mathbb{I})\right)
\subseteq \llbracket S\rrbracket_\mathbb{I}\left(\llbracket\sigma_0 \rrbracket_\mathbb{I}^\ast(\llbracket\gamma\rrbracket_\mathbb{I})\right)\rightarrow\llbracket S\rrbracket_\mathbb{I}\left(\llbracket\sigma_0 \rrbracket_\mathbb{I}^\ast(\llbracket\delta \rrbracket_\mathbb{I})\right).
\end{align*} Here, the last inclusion is derived by Lemma A.5 (4). It follows that 
\begin{equation}\label{adapt-6}\begin{split}&\llbracket\gamma \rrbracket_\mathbb{I} \cap \llbracket S\rrbracket_\mathbb{I}(\llbracket(\forall\overline{p})(\gamma\rightarrow\delta) \rrbracket_\mathbb{I})=\llbracket S\rrbracket_\mathbb{I}\left(\llbracket \sigma_0\rrbracket_\mathbb{I}^\ast(\llbracket\gamma \rrbracket_\mathbb{I})\right) \cap \llbracket S\rrbracket_\mathbb{I}(\llbracket(\forall\overline{p})(\gamma\rightarrow\delta) \rrbracket_\mathbb{I})\\ &\subseteq \llbracket S\rrbracket_\mathbb{I}\left(\llbracket \sigma_0\rrbracket_\mathbb{I}^\ast(\llbracket\gamma \rrbracket_\mathbb{I})\right) \cap
\left[\llbracket S\rrbracket_\mathbb{I}\left(\llbracket\sigma_0 \rrbracket_\mathbb{I}^\ast(\llbracket\gamma\rrbracket_\mathbb{I})\right)\rightarrow\llbracket S\rrbracket_\mathbb{I}\left(\llbracket\sigma_0 \rrbracket_\mathbb{I}^\ast(\llbracket\delta \rrbracket_\mathbb{I})\right)\right]\\
&\subseteq \llbracket S\rrbracket_\mathbb{I}\left(\llbracket\sigma_0 \rrbracket_\mathbb{I}^\ast(\llbracket\delta \rrbracket_\mathbb{I})\right)=\llbracket\delta \rrbracket_\mathbb{I}.
\end{split}\end{equation} 
Note that the second inclusion in (\ref{adapt-6}) is derived by the fact that the Sasaki implication $\rightarrow$ satisfies $a\wedge (a\rightarrow b)\leq b$ (see Remark \ref{remark-Sasaki}), which is in turn guaranteed by the Ortho-modularity (see Theorem \ref{thm-Sasaki}). Substituting (\ref{adapt-6}) into (\ref{adapt-5}), we obtain:
$\llbracket S\rrbracket_\mathbb{I}\left(\llbracket\tau \rrbracket_\mathbb{I}^\ast(\llbracket\beta \rrbracket_\mathbb{I}\cap \llbracket(\forall\overline{p})(\gamma\rightarrow\delta) \rrbracket_\mathbb{I})\right)\subseteq \llbracket\tau \rrbracket_\mathbb{I}^\ast\left(\llbracket\delta \rrbracket_\mathbb{I} \right)=\llbracket\delta \rrbracket_\mathbb{I}$ because $\mathit{var}(\tau)\subseteq\overline{q}=\mathit{free}(\beta)\cup\mathit{free}(\gamma)\setminus[\mathit{free}(\delta)\cup\overline{p}]$ and thus $\mathit{var}(\tau)\cap\mathit{free}(\delta)=\emptyset.$ Thus, we compete the proof.
\end{proof}


\subsection{Quantum Assertions in Runtime Checking}

Let us now turn to consider a practical application. As discussed in the Introduction, a runtime checking scheme was developed in \cite{Liu21} to assert that a program variable is in the subspace spanned by a set of quantum states, say $|\psi_1\rangle,...,|\psi_n\rangle$. Unfortunately, this scheme was interpreted incorrectly as to assert that the program variable is in one of states $|\psi_1\rangle,...,|\psi_n\rangle$ due to the misunderstanding of logical connective (disjunction) in quantum logic. By adopting the formally defined Birkhoff-von Neumann quantum logic as an assertion logic, as suggested in this paper, such a confusion between quantum disjunction and its classical counterpart can be easily avoided.   

\begin{exam}\label{q-connectives} Let $\hs$ be the Hilbert space of a quantum program variable $q$. For each $|\psi\rangle\in\hs$, we write $[\psi]$ for the atomic proposition that variable $q$ is in the $1$-dimensional closed subspace of $\hs$ generated by the single state $|\psi\rangle$. Then the proposition that $q$ is in the closed subspace generated by a family of states $|\psi_1\rangle,...,|\psi_n\rangle$ can be written as $\bigvee_{i=1}^n[\psi_i]=[\psi_1,...,\psi_n].$ 
Moreover, if $\{|\psi_0\rangle,|\psi_1\rangle,...,|\psi_n\rangle\}$ is an orthonormal basis of $\hs$, then \begin{equation}\label{example-negation}[\psi_1,...,\psi_n]=\neg[\psi_0],\end{equation} and more generally, $[\psi_{k},\psi_{k+1},...,\psi_n]=\neg \bigvee_{i=0}^{k-1}[\psi_i]$ for any $1\leq k<n$.    
\end{exam}

The benefit of employing Birkhoff-von Neumann quantum logic as an assertion logic for quantum programs is not limited to precise specification of quantum assertions. Indeed, it also often provides a more economic way to specify quantum assertions. For example, as pointed out in the Introduction, if $q$ denotes an $N$-qubit system, then without a logical language the quantum assertion (\ref{example-negation}) needs to be stored as a $2^N\times 2^N$ matrix in implementation. But the logical representation given as the right-hand side of (\ref{example-negation}) is much more compact. 

Moreover, the adaptation rules in Figure \ref{fig auxi-rules} are often helpful for quantum assertion checking. For example, suppose we want to assert that the output of a quantum program $S$ is always in a subspace $Y$ of its state space $\hs$ for all inputs from a subspace $X$ of $\hs$; that is, $\models_\mathit{par}\{X\}S\{Y\}$. We choose a basis $|\psi_1\rangle,...,|\psi_n\rangle$ of $X$. Then $X=\bigvee_{i=1}^n[\psi_i]$, and by rule (Disjunction), it suffices to check $\models_\mathit{par}\{[\psi_i]\}S\{Y\}$ for $i=1,...,n$.   
We even expect that automatic tools for quantum assertion checking can be implemented based on the logical mechanism developed in this paper. 



\section{Conclusion}\label{sec-Con} 

In this paper, we defined an extension of Birkhoff-von Neumann quantum logic, namely, a first-order logic $\mathcal{QL}$ with quantum variables, as an assertion language for quantum programs. In particular, $\mathcal{QL}$ was incorporated into quantum Hoare logic (QHL) so that the relative completeness of QHL can be formulated in a more formal way than that in the previous literature, and a series of adaptation rules can be derived to ease the verification, analysis and runtime checking of quantum programs. But several interesting problems about $\mathcal{QL}$ itself as well as the combination of $\mathcal{QL}$ and QHL are still unsolved:   

{\vskip 3pt}

\textbf{More Quantifiers over Quantum Variables}: The quantification over quantum variables in $\mathcal{QL}$ is defined by allowed quantum operations on the quantified variables (see equation (\ref{quantum-quantifier}), clause (5) in Definition \ref{QFO-semantics} and clause (4) in the definition of the semantics of QT$_=$ in Subsection \ref{subsec-equational}). But there are some other interesting ways to introduce quantifiers in $\mathcal{QL}$. For example, a universally quantified formula $(\forall \overline{q})\beta$ with quantum variables $\overline{q}$ can be interpreted according to different levels of the correlation between $\overline{q}$ and other quantum variables:   
\begin{itemize}\item \textit{Product quantification}: $(\mathbb{I},\rho)\models (\forall_p\overline{q})\beta\ {\rm iff}\  (\mathbb{I}, \sigma\otimes\rho\downarrow(\mathit{Var}\setminus\overline{q})\models\beta\ {\rm for\ any}\ \sigma\in\mathcal{D}(\hs_{\overline{q}}).$
\item \textit{Separation quantification}: $(\mathbb{I},\rho)\models (\forall_s\overline{q})\beta\ {\rm iff}\ \left(\mathbb{I},\sum_i(\sigma_i\otimes\rho_i)\right) \models\beta\ {\rm for\ any}\ \sigma_i\in\mathcal{D}(\hs_{\overline{q}})\ {\rm and}\ \rho_i\in\mathcal{D}(\hs_{\mathit{Var}\setminus\overline{q}}) {\rm with}\ \sum_i\rho_i=\rho\downarrow (\mathit{Var}\setminus\overline{q}).$
\item \textit{Entanglement quantification}: $(\mathbb{I},\rho)\models (\forall_e \overline{q})\beta\ {\rm iff\ for\ any}\ \rho^\prime\ {\rm with}\ \rho^\prime\downarrow(\mathit{free}(\beta)\setminus\overline{q})=\rho\downarrow(\mathit{free}(\beta)\setminus\overline{q}), (\mathbb{I},\rho^\prime)\models\beta.$
\end{itemize}
 Here, $\mathit{Var}$ denotes the set of all quantum variables, $\otimes$ stands for tensor product, $\rho\downarrow X$ is the restriction of a quantum state on a subset $X\subseteq\mathit{Var}$ of quantum variables. 

Obviously, the above three quantifications over quantum variables and the one studied in this paper are useful in different circumstances, and an extension of $\mathcal{QL}$ with these new quantifiers can serve as a stronger logic tool for reasoning about quantum computation and quantum information. 
For example, it can help to deal with ghost (auxiliary) quantum variables considered in \cite{Unruh19b}. At the same time, a series of new problems arise in this new logic; in particular, quantifier elimination, which will be closely connected to some fundamental issues about correlation between quantum systems, we believe,  one way or another. 

{\vskip 3pt}

\textbf{Assertion Languages for Other Quantum Program Logics}: 
$\mathcal{QL}$ is designed as an assertion logic for quantum Hoare logic (QHL). Several extensions of QHL has been proposed in the literature, including relational quantum Hoare logic (qRHL) \cite{Unruh19a,Li21,Barthe20} and quantum separation logic (QSL) \cite{Zhou21,Le22}. However, the assertion languages for all of these quantum program logics have not been formally defined. It seems that $\mathcal{QL}$ can be directly used as an assertion language for qRHL, but it is not the case for QSL. Some useful quantum generalisations of separation conjunction and implication were introduced in \cite{Le22,Zhou21}. But we believe that more research on QSL is needed in order to find quantum separation connectives with the presence of entanglement in a more serious consideration. Furthermore, it would be nice to define them in a formal logical language so that an expansion of $\mathcal{QL}$ with them can serve as an assertion language of QSL.

{\vskip 3pt}

\textbf{Adding Classical Variables}: For simplifying the presentation, quantum Hoare logic (QHL) was originally designed in \cite{Ying11} for purely quantum programs without classical variables. This simplification does not reduce its expressive power because classical computation can be simulated by quantum computation. In practical applications, however, it is often much more convenient to handle quantum variables and classical variables separately. 
The \textbf{while}-language with both classical and quantum variables is introduced in \cite{Ying11a} where a correctness formula (Hoare triple) is defined with the pre/postcondition as a pair of a classical first-order logical formula and a quantum predicate so that the former specifies the properties of classical variables and the latter for quantum variables.      
Later, a QHL with both quantum and classical variables was developed in \cite{Feng21}.
A limitation of the logic in \cite{Feng21} is that preconditions and postconditions are defined to be so-called classical-quantum predicates, each of which is represented as a family of Hermitian operators indexed by the states of classical variables. Such a representation is quite cumbersome, and should cause the issue of (double) explosion of the state spaces of both classical and quantum variables.  
In particular, the compactness offered by a logical language, even for classical variables, is totally lost.  
We believe that an assertion logic can significantly simplify reasoning about quantum algorithms with the program logic in \cite{Feng21} and thus improve its applicability.   
This then requires us to combine the original first-order quantum logic QL with classical variables (see Subsection \ref{Sec-QL-FO}) and $\mathcal{QL}$ with quantum variables newly introduced in this paper into a single logic system. It seems that the correctness formulas defined in \cite{Ying11a} are more convenient that those in \cite{Feng21} for this purpose.    

{\vskip 5pt}
 
\textbf{Acknowledgment}:  The author would like to thank Prof. Yuan Feng and Dr. Li Zhou for useful discussions. This work was partially supported by the National Key R\&D Program of China (2018YFA0306701) and the National Natural Science Foundation of China (61832015).

\end{document}